\documentclass [12pt]{article}

\usepackage{amsmath,amsthm,amsfonts,amscd,latexsym,amssymb,amsbsy,dsfont,mathrsfs,color}
\usepackage[small, centerlast, it]{caption}
\usepackage{epsfig}
\usepackage[titles]{tocloft}
\usepackage[latin1]{inputenc}
\usepackage{verbatim} 
\usepackage[active]{srcltx} 
\usepackage{fancyhdr}
\usepackage{caption}
\usepackage{enumerate}
\usepackage{color}
\usepackage{hyperref}

\definecolor{NiColor}{RGB}{77,77,255}
\definecolor{NiColoRed}{RGB}{255,77,77}
\definecolor{NiCitation}{RGB}{77,255,77}

\newsymbol\rest 1316    
\def\emptyset{\varnothing} 
\def\b1{{1\!\!1}}

\def\cB{\mathscr{B}}

\def\cD{{\ca D}}
\def\cF{\mathscr{F}}

\def\cL{\mathscr{L}}
\def\cM{\mathscr{M}}

\def\cO{{\ca O}}

\def\cS{\mathscr{S}}

\def\sH{{\mathsf H}}

\def\bC{{\mathbb C}}           

\def\bE{{\mathbb E}}

\def\bN{{\mathbb N}}

\def\bR{{\mathbb R}}
\def\bP{{\mathbb P}}

\def\gB{{\mathfrak B}}

\def\beq{\begin{eqnarray}}
\def\eeq{\end{eqnarray}}

\newcommand{\ca}[1]{{\cal #1}}         

\usepackage{sectsty}
\sectionfont{\normalsize}
\subsectionfont{\normalsize\normalfont\itshape}
\newtheoremstyle{plain}
{5pt}
{9pt}
{\itshape}
{}
{\itshape\bfseries}
{}
{1em}
{}

\theoremstyle{thm}
\newtheorem{theorem}{\em Theorem}[section]
\newtheorem{lemma}[theorem]{\em Lemma}

\newtheorem{proposition}[theorem]{\em Proposition}
\newtheorem{definition}[theorem]{\em Definition}
\newtheorem{remark}[theorem]{\em Remark}
\newtheorem{example}[theorem]{\em Example}

\begin{document}


\par
\bigskip
\large
\noindent
{\bf  An operational construction of the sum of two non-commuting observables in quantum theory and related constructions}
\bigskip
\par
\rm
\normalsize


\noindent  {\bf Nicol\`o Drago$^{a}$}, {\bf Sonia Mazzucchi$^{b}$},  {\bf Valter Moretti$^{c}$ (Corresponding author)}\\
\par

\noindent 
  Department of  Mathematics, University of Trento, and INFN-TIFPA \\
 via Sommarive 14, I-38123  Povo (Trento), Italy.\\

\noindent $^a$nicolo.drago@unitn.it,  $^b$sonia.mazzucchi@unitn.it,  $^c$valter.moretti@unitn.it\\

 \normalsize

\par

\rm\normalsize

\noindent {\small August 2020}

\rm\normalsize


\par
\bigskip

\noindent
\small
{\bf Abstract}.   The existence of a real linear-space structure on  the set of observables of a quantum system -- i.e., the requirement that
 the linear combination of two generally non-commuting observables $A,B$ is an observable as well -- is a fundamental 
postulate  of  the quantum theory yet before introducing any  structure of algebra. 
However, it is by no means clear how to choose the measuring instrument of a general observable of the form
 $aA+bB$ ($a,b\in \bR$) if such measuring instruments are given for the addends observables $A$ and $B$
 when they are  incompatible observables. A mathematical version of this dilemma is how to construct the 
spectral measure of $f(aA+bB)$ out of the spectral measures of $A$ and $B$.
We present such a construction with a formula which is valid for general unbounded selfadjoint operators 
$A$ and $B$, whose spectral measures may not commute,  and a wide class of  functions $f: \bR \to \bC$.  
In the bounded case,  we prove that  the Jordan product of $A$ and $B$ (and suitably symmetrized polynomials of $A$ and $B$) can be constructed with the same
 procedure out of the spectral measures of $A$ and $B$. The formula turns out to have an interesting 
 operational interpretation and, in particular cases, a nice interplay with the theory of Feynman path  
integration and the Feynman-Kac formula.
\normalsize

\tableofcontents

\section{Introduction}\label{Sec: introduction} 

The elementary formulation of Quantum Theory  in a {\em complex separable} Hilbert  space $\sH$ can be described  as a {\em non-Boolean} probability theory. There, 
{\em quantum states} $\mu$ are {\em probability measures} over the orthomodular separable  $\sigma$-complete lattice
 $\cL(\sH)$ of orthogonal projectors (the order relation being the standard inclusion of projection subspaces)
 which generalizes the notion of $\sigma$-algebra (see, e.g., \cite{vonNeumann,Mackey,varadarajan,BC,landsman,Moretti2017}).
 The Hilbert space $\sH$ depends on the physical system $S$ one intends to study  and the elements of $\cL(\sH)$ physically represent {\em elementary propositions} of $S$,  also known as {\em YES-NO elementary observables} 
or {\em tests} about $S$. They  are therefore assumed to have a definite operational meaning in terms of experimental devices.
{\em Commutativity} of  $P,Q \in \cL(\sH)$ has the physical meaning of {\em  simultaneous measurability}
 of the elementary observables  represented by $P$ and $Q$. In that case  $P$
and $Q$ are said to be {\em compatible}.

%


The issue  we want to address in this work concerns the set $\cO_S$ of   {\em observables} of $S$.
An {\em observable} $A$ is  primarily a collection of elementary propositions $P^{(A)}_E$ labelled by real 
 Borel sets $E\in \cB(\bR)$.  The operator $P^{(A)}_E$ admits an explicit operational interpretation in terms of the statement ``the outcome of the measurement of the observable $A$ belongs to $E$.''  We stress that there must be one (or several)  measuring instrument(s) associated to the class $\{P^{(A)}_E\}_{E \in \cB(\mathbb{R})}$.
Sound physical and logical arguments  (see, e.g., \cite{varadarajan,landsman,Moretti2017}) require that the collection 
$\{P^{(A)}_E\}_{E \in \cB(\mathbb{R})}$ define a $\sigma$-Boolean algebra homomorphism form $\cB(\bR)$ to $\cL(\sH)$. In terms more familiar to physicists,  $\{P^{(A)}_E\}_{E \in \cB(\mathbb{R})}$ is a {\em projector-valued measure}
 (PVM) also known as {\em spectral measure}.
Among them,   {\em compatible} observables $A$, $B$ are by definition those whose associated PVMs are made of pairwise compatible projectors, i.e.,  $P_E^{(A)}P_F^{(B)}=P_F^{(B)}P_E^{(A)}$ for all $E,F \in \cB(\bR)$.

At this juncture,  the spectral machinery establishes that there is a one-to-one correspondence between PVMs over $\sH$  and 
(generally unbounded) selfadjoint operators  $A: D(A) \to \sH$.
As a matter of fact,  the following holds
\begin{equation}
	A = \int_{\bR} \lambda dP^{(A)}(\lambda)\:.\label{st}
\end{equation}
We conclude that $\cO_S$  is made of {\em some} selfadjoint operators $A: D(A) \to \sH$ and  the natural question emerges about how large
$\cO_S$ is  in  the whole  set of selfadjoint operators  over $\sH$.

As a matter of fact,  it is generally assumed {\em as a further postulate}  that  $\cO_S^{(b)}:=\cO_S\cap\mathfrak{B}(\mathsf{H})$  is  a linear subspace
\footnote{It is finally assumed that, in absence of {\em superselection rules} and {\em gauge symmetries}, the complex span of the elements in $\cO_S^{(b)}$ amounts to the whole $\gB(\sH)$, otherwise it is a von Neumann algebra in $\sH$.}
of the real linear space of bounded
\footnote{
	Notice that, $\cO_S^{(b)}$ suffices to recover all physical information carried by $\cO_S$.
	Indeed for all $A\in\mathcal{O}_S$ one has $A_N := \int_{[-N,N]} \lambda dP^{(A)}(\lambda)\in\mathfrak{B}(\mathsf{H})$ ($N\in\mathbb{N}$) as well as $(A_N-A)\psi\to 0$ as $N\to+\infty$ for all $\psi\in D(A)$.
	The operator $A_N$ can be interpreted as the same observable $A$ measured with instruments whose range is restricted to $[-N,N]$.}
selfadjoint operators of $\gB(\sH)$.
This structure is  later enriched \cite{Emch,strocchi,landsman,Moretti2017})  by adding  in particular the so-called (non-associative) {\em Jordan  product}, 
 and assuming suitable (weak) topological features compatible with the spectral machinery. The final construction turns out to be a concrete
 {\em Jordan $W$-algebra} in   $\gB(\sH)$ or, more generally, a concrete {\em von Neumann algebra}  in   $\gB(\sH)$ when complex combination
 of observables are permitted.  These structures  are actually difficult to physically  justify {\em a priori}. If assumed, they however promote
 the theory to a very high level of effectiveness in physics {\em a posteriori}.

Already sticking to the linear structure of the set of bounded observables, we argue that a long standing (see, e.g., \cite{gudder,streater}) open issue  pops out, with a twofold nature, both  physical and mathematical.
Suppose that  $A,B \in \cO_S^{(b)}$ and that we know their  respective   measuring instruments.
If  we pick out $a,b \in \bR$, in spite of the postulate about the linear structure of $\cO_S^{(b)}$,  there is no general way  to associate a measuring instrument to $aA+bB$ out of those of $A$ and $B$.
This obstruction is valid unless the observables $aA$ and $bB$ are compatible
\footnote{
	The problem tackled in this work is a consequence of the existence of incompatible observables so that it does {\em not} arise in classical  physics.
	Given a classical system $S$, for instance described in its phase space, there is a minimal set of observables such that all remaining observables are functions of them.
	Therefore, measuring the first ones exactly amounts  to measuring {\em all} observables of $S$.
}.
In that case, the instruments corresponding to the {\em joint PVM} of $P^{(aA)}$ and $P^{(bB)}$ can be exploited.

The mathematical version of the  issue raised above is that {\em no}  general formula is known  determining the operator  $f(aA+bB)$ -- and every  projector
$P^{(aA+bB)}_E$   in particular -- as a function of  $P^{(A)}$ and $P^{(B)}$ when these PVMs do not commute.
This work addresses  the outlined problem from the mathematical side.
We will prove -- \textit{cf.} (\ref{mainid2}) and theorem \ref{teorapgen} -- that for a suitable class of continuous functions $f: \bR \to \bR$ the following formula holds true:
\begin{equation}\label{formula1}
	f\left(\overline{aA+bB}\right)
	=\mbox{s-}\lim_{N\to +\infty} \int_{\bR^{2N}} f\left( \frac{1}{N}\sum^N_{n=1}  
	(a\lambda_{n} + b\mu_{n})\right) dP^{(A)}_{\lambda_1}dP^{(B)}_{\mu_1}\cdots dP^{(A)}_{\lambda_N}dP^{(B)}_{\mu_N}\:.
\end{equation}
This way the selfadjoint operator $f(aA+bB)$ can be computed out of the PVMs of the selfadjoint operators $A$ and $B$ .
Equation \eqref{formula1} is valid for generally non-commuting $P^{(A)}$ and $P^{(B)}$,  also in the case of unbounded selfadjoint  operators $A$ and $B$ under the condition that $aA+bB$ is essentially selfadjoint over $D(A)\cap D(B)$.
The class of functions $f$ is sufficiently large to include all polynomials when $A$ and $B$ are bounded, in particular $aA+bB$ itself. 

As one of the various  byproducts of the above-mentioned result for the case of $A$ and $B$ bounded, we also find an identity connecting the {\em Jordan product} $A\circ B := \frac{1}{2}(AB+BA)$ to the generally non-commuting PVMs $P^{(A)}$ and $P^{(B)}$ -- \textit{cf.} theorem \ref{teopolynomials}.

\paragraph{Comparison with existing literature.}
The assumption of the existence of a real linear space structure over the set of observables is crucial also in much more abstract, and apparently more operational, formulations of the quantum theory from \cite{Mackey}to \cite{strocchi}.
Therein the primary object is an abstract {\em Jordan Banach $^*$-algebra} (or directly a $C^*$-algebra), out of which the Hilbert space formulation is recovered a posteriori through the GNS construction.
Observables are defined in terms of their {\em expectation values} (here properly interpreted as {\em states}) and it is {\em postulated} that, given a pair of observables $A$ and $B$ and a corresponding pairs of reals $a,b$,  there is an observable  denoted by $aA+bB$ whose expectation values are the linear combination of the expectation values of $A$ and $B$ with  coefficients $a$ and $b$.
Though this observable is proved to be unique (because the states are reasonably supposed to separate the observables), no ideas are supplied to solve the problem
we pointed out: how we can measure $aA+bB$ if we know how to measure $A$ and $B$ when they are incompatible.

From the mathematical side our results are entangled with the construction of a functional calculus for non-commuting operators.
Several rigorous results exist in the literature about this subject.
However, to the best of our knowledge, almost all are concentrated on two physically well-motivated approaches: the generalization of {\em Feynman's functional calculus for $T$-ordered products}  of operators and extensions of {\em Weyl calculus} -- see \cite{JL,JoLaNi,GiZa} and \cite{NSS} for exhaustive reviews on these approaches -- with several original contributions from outstanding authors \cite{Nelson,Araki,Maslov}.
The mathematical technology, the types of functional spaces and the interactions with some other mathematical objects like the Feynman-Kac integral  share  some similarities with the content of in this work -- \textit{cf.} section \ref{Sec: relation with Feynman integration}.
However, to authors' knowledge, none of those papers  presents   an integral expansion like  our formula (\ref{formula1})  where (a) a joint integral with respect to the non-commuting spectral measures of the involved operators shows up and (b)  the physically suggestive averaged value of spectral parameters takes place -- \textit{cf.} section \ref{Sec: physical interpretation}.

\paragraph{Structure of the work.}
The paper is organized as follows.
After recalling some general facts and the general notation used throughout,  in section \ref{Sec: extension to infinite dimensional Hilbert spaces} 
we  tackle the problem raised above first in the finite-dimensional case and next in the infinite dimensional case.
Dealing with a pair of selfadjoint operators $A$, $B$ in a complex Hilbert space and a continuous bounded function $f: \bR \to \bC$ of a certain class,
our main result consists of a formula where the operator $f(aA+bB)$ is written as the (strong) limit of a sequence  of certain operator-valued integrals 
 whose measure is an alternating  product of spectral measures of $A$ and $B$.
As a matter of fact, we prove  Eq. (\ref{prototypeC}) in the finite-dimensional case and Eq. (\ref{mainid2}) in the much harder infinite-dimensional case.
Section \ref{Sec: some applications} is devoted to study some applications: In particular we present a (fairly explicit) construction for the PVM of $aA+bB$ and the Jordan product $A\circ B$ in terms of the PVM of $A$ and $B$.
Eventually we discuss the nice interplay with both the theory of Feynman integration in the phase space as well as the so called Feynman-Kac formula.
Section \ref{Sec: counterexamples} contains some counterexamples which demonstrate how part of the structure introduced in the
 infinite dimensional case cannot be improved. 
In section \ref{Sec: physical interpretation} we discuss the physical interpretation of the constructed formalism in the finite dimensional case where no technical issues take place, allowing us to focus directly on physics. A short summary ends the paper.

\paragraph{General conventions and notation.}
We adopt throughout the paper  the standard definition of {\em complex measure} \cite{Rudin87}, as a map  $\mu :\Sigma \to \bC$ 
which is  {\em unconditionally $\sigma$-additive}  over the $\sigma$-algebra $\Sigma$.
With this definition the {\em total variation} $||\mu|| :=|\mu|(\Sigma)$ turns out to be  finite. 
By $\cB(X)$ we denote the Borel $\sigma$-algebra over a topological space $X$ and $L^p(X, \nu)$ denotes the standard  $L^p$ space with respect to a positive 
{\em Borel measure} $\nu$   \cite{Rudin87} (in case $\nu$ is complex, the associated $L^p$ spaces are those referred to $|\nu|$).
By $L^p(\bR, dx)$  (also with $dy$ in place of $dx$) we mean the $L^p$ space with respect to  Lebesgue measure (viewed as a Borel measure) 
on $\bR$.
In case the measure $\nu$ is defined on a generic $\sigma$-algebra $\Sigma$  over the set $X$, we use the more precise notation $L^p(X, \Sigma, \nu)$. 
If $E\subset X$, for any given set $X$, then the {\em  indicator function} $1_E$ of $E$ is defined as $1_E(x):= 1$ if $x\in E$ and $1_E(x)=0$ if 
$x\in X \setminus E$.

A {\em Hilbert space} is always assumed to be {\em complex}. We assume the usual convention concerning  {\em standard domains} of
 composition of operators over a Hilbert space $\sH$,
$A:D(A) \to \sH$ and $B: D(B) \to \sH$ with $D(A), D(B) \subset \sH$ linear subspaces:
(i) $D(BA) := \{\psi \in D(A) \:|\: A\psi \in D(B)\}$;
(ii) $D(A+B):= D(A)\cap D(B)$;
(iii) $D(aA):=D(A)$ for $a\in \bC \setminus \{0\}$ and $D(0A):= \sH$.
These requirements yield in particular  $D(aA+bB) = D(A)\cap D(B)$ if $ab \neq 0$.
We define $A^k := A \cdots (\mbox{$k$ times})\cdots A$
if $k=1,2,\ldots$ and $A^0:=I$. By $A\subset B$ we mean that $D(A) \subset D(B)$ and $B|_{D(A)}=A$. 

We denote by $A^*:D(A^*) \to \sH$ the {\em adjoint operator} of a densely-defined operator $A: D(A) \to \sH$  and by $\overline{A}$ its {\em closure}.  
A densely-defined operator $A$ is {\em selfadjoint} if $A^*=A$. 
The operator $A$ is {\em essentially selfadjoint} if $A\subset A^*$ and $\overline{A}$ is selfadjoint.
We denote by $\gB(\sH)$  the unital $C^*$-algebra of bounded everywhere defined operators over the Hilbert space $\sH$, and $\cL(\sH)\subset \gB(\sH)$
 denotes the  lattice of {\em orthogonal 
projectors}: $P=P^*=PP$.

The {\em spectrum} of an operator $A$ is always denoted by $\sigma(A)$.
According to the {\em spectral theory} (details, e.g., in  \cite{RS2,Moretti2017}), if $T: D(T) \to \sH$ is 
a generally unbounded selfadjoint operator in $\sH$ and $g: \bR \to \bC$ is measurable, then one has
\begin{align}\label{Eq: function of self-adjoint operator}
	g(T):= \int_{\bR} g(\tau)  \: dP^{(T)}(\tau)\,,\qquad
	D(g(T)) := \left\{\psi \in \sH \:\left|\: \int_\bR |g(\tau)|^2 d \mu^{(T)}_{\psi,\psi}(\tau) <+\infty\right.\right\}\,.
\end{align}
Above,  $ \mu^{(T)}_{\phi, \psi}(E) :=
\langle \phi|P^{(T)}_E\psi \rangle$ if $E\in \cB(\bR)$ defines a  regular complex Borel measure, which is positive  if $\psi=\phi$,
 and  $\{P^{(T)}_E\}_{E\in \cB(\bR)}$ is the {\em projection-valued measure} (PVM) -- also known as a  {\em spectral measure} -- uniquely associated to $T$.
The right-hand side of (\ref{Eq: function of self-adjoint operator}) can be defined \cite{Rudin93}\footnote{Alternatively, (\ref{spectral2}) 
is a consequence of a different but equivalent definition of the right-hand side of  (\ref{Eq: function of self-adjoint operator}), 
which uses the strong operator topology instead of the weak one. See, e.g., \cite{Moretti2017}}  as  the unique operator $T$ on $D(g(T))$ such that
\begin{equation}\langle \phi| g(T)\psi \rangle = \int_{\bR} g(\tau) d\mu_{\phi,\psi}(\tau) \quad 
\mbox{for all $\phi \in \sH$ and $\psi \in D(g(T))$.}\label{spectral2}\end{equation}
The support of $P^{(T)}$  coincides with the spectrum $\sigma(T)$ and the integration in (\ref{Eq: function of self-adjoint operator})
 can be restricted accordingly without affecting the identity. The {\em spectral decomposition} of $T$ is nothing but
 (\ref{Eq: function of self-adjoint operator}) with $g : \bR \ni \tau \mapsto \tau \in \bR$ and the {\em spectral theorem} 
  just  states that this decomposition exists and $P^{(T)}$ is uniquely determined by $T$.

\section{From $\bC^n$ to  an infinite dimensional Hilbert space}\label{Sec: extension to infinite dimensional Hilbert spaces}

\subsection{The finite-dimensional Hilbert space case}\label{finite}
The finite-dimensional case is a comfortable  arena where building up our formalism, since no problems related 
with domains and choices of topologies take place.
Let us therefore suppose that $\sH$ is finite-dimensional.
Without loss of generality, we can assume $\sH= \bC^k$.  From the elementary spectral theory \cite{Schmu}, we define
\begin{equation}g(T) := \sum_{\tau \in \sigma(T)} g(\tau) P_\tau^{(T)}\label{elemspectr}\end{equation}
for every function $g: \bR \to \bR$ and every selfadjoint operator $T$, where $T$  is  a $k\times k$  Hermitian matrix
 $T: \sH \to \sH$.  The spectrum $\sigma(T)\subset \bR$ of $T$ denotes its set of eigenvalues and $P^{(T)}_\tau$ indicates
 the orthogonal projector onto the eigenspace of $\tau \in \sigma(T)$, so that $\{P_\tau^{(T)}\}_{\tau\in \sigma(T)}$
 completely determines the PVM $P^{(T)}$ of $T$.
No technical snags pop out with the sum in (\ref{elemspectr}), since $\sigma(T)$ is a non-empty  set of at most 
 $k= \dim (\sH) <+\infty$ elements.

We now focus on a couple of selfadjoint operators $A: \sH \to \sH$, $B: \sH\to \sH$ with, in general,  $AB \neq BA$. We intend to write $f(aA+bB)$
 in terms of the PVMs $P^{(A)}$ and $P^{(B)}$ for $a,b\in \bR$ and where 
$f: \bR \to \bC$ is a sufficiently regular function.  For technical reasons which will be evident shortly, we henceforth assume that:
(i) $f : \bR \to \bC$ is continuous;
(ii) $f$ belongs to $L^1(\bR, dx)$;
(iii) the inverse Fourier transform $\check{f}(t) := \frac{1}{2\pi}\int_\bR f(x) e^{-itx} dx$ belongs to
$L^1(\bR, dx)$.
Every $f\in \cS(\bR)$ in particular satisfies the said three hypotheses.
Assuming (i),(ii) and (iii),  the {\em pointwise inversion theorem} of the Fourier transform is valid: $f(x) = \int_\bR \check{f}(t) e^{itx} dt$ for 
every $x\in \mathbb{R}$.

The pivotal tool underpinning our result is the celebrated {\em Trotter formula},
\begin{equation}e^{it(aA+bB)} = \lim_{N\to +\infty} \left(e^{i\frac{t}{N}aA} e^{i\frac{t}{N}aB}\right)^N\:, \quad t \in \bR\:\label{Trotter0},\end{equation}
whose convergence is  here referred to every normed topology on $\gB(\sH)$, since that (Banach) space is finite dimensional.
Spectrally decomposing  the exponentials on both sides into a finite linear combinations of matrices according to (\ref{elemspectr}) and using
  of Lebesgue's dominated 
convergence theorem 
term by term 
while  integrating both sides against  $\check{f}(t)$,  we have, 
\begin{equation}\label{6-bis}\int_\bR\check{f}(t) e^{it(aA+bB)} dt = \lim_{N\to +\infty} \int_{\bR} \check{f}(t) \underbrace{e^{i\frac{t}{N}aA} e^{i\frac{t}{N}aB}\cdots 
e^{i\frac{t}{N}aA} e^{i\frac{t}{N}aB}}_{N \: \mathrm{times}} dt\:.\end{equation}
Exploiting (\ref{elemspectr}), the  identity \eqref{6-bis} can be expanded as
\begin{align}&\sum_{\tau \in \sigma(aA+bB)}\left(\int_\bR\check{f}(t) e^{it\tau} dt \right)\: P_\tau^{(aA+bB)} \nonumber \\&=
 \lim_{N\to +\infty} \sum_{\lambda_j \in \sigma(A), \mu_j\in \sigma(B)} \left(\int_{\bR} \check{f}(t) e^{i\frac{t}{N}a\lambda_1} e^{i\frac{t}{N}b\mu_1}\cdots e^{i\frac{t}{N}a\lambda_N} e^{i\frac{t}{N}b\mu_N} dt \right) P^{(A)}_{\lambda_1} 
 P^{(B)}_{\mu_1}\cdots P^{(A)}_{\lambda_N}  
P^{(B)}_{\mu_N}\:.\nonumber \end{align}
Since $f(x) = \int_\bR \check{f}(t) e^{itx} dt$  and taking (\ref{elemspectr}) into account on the left-hand side, we  obtain the 
prototype of our general formula
\begin{equation}\label{Eq: f(aA+bB) in the finite dimensional case}
	f(aA+bB)=
	\lim_{N\to +\infty}\sum_{\lambda_j \in \sigma(A), \mu_j\in \sigma(B)}
	f\left( \frac{1}{N}\sum_{j=1}^N (a\lambda_j+ b\mu_j)\right)
	P^{(A)}_{\lambda_1}P^{(B)}_{\mu_1}\cdots P^{(A)}_{\lambda_N}P^{(B)}_{\mu_N}\:.
\end{equation}
 In  $\sH= \bC^k$, it is not difficult to prove that $AB-BA=0$  implies
$\big(e^{i\frac{t}{N}aA} e^{i\frac{t}{N}aB}\big)^N = e^{it(aA+bB)}\:. $ Hence,  following our derivation of (\ref{Eq: f(aA+bB) in the finite dimensional case}),
the limit on the right-hand side of (\ref{Eq: f(aA+bB) in the finite dimensional case}) is not necessary since its 
 argument  turns out to be constant in $N$.
Therefore, for commuting $A$ and $B$, equation (\ref{Eq: f(aA+bB) in the finite dimensional case}) reduces to
\begin{equation}\label{prototypeC}f(aA+bB)=\sum_{\lambda \in \sigma(A), \mu\in \sigma(B)} f( a\lambda+ b\mu) P^{(A)}_{\lambda} 
 P^{(B)}_{\mu}\:,\end{equation} 
which is nothing but the usual spectral decomposition of the left-hand side with respect to the {\em joint PVM} of $A$ and $B$.

A physically important  case of (\ref{prototypeC}) would be  obtained by choosing $f= 1_{[\alpha,\beta)}$  (the indicator function of the set $[\alpha,\beta)$). 
As a matter of fact,   $1_{[\alpha,\beta)}(aA+bB) = P^{(aA+bB)}_{[\alpha,\beta)}$ is the spectral projection of $[\alpha,\beta)$ and this family of 
projectors, when $\alpha < \beta$,  includes the full information of the PVM of $aA+bB$. Unfortunately such $f$ does not satisfy (i) and (iii), 
hence a further regularization procedure  is necessary by means of a family $1^{(\epsilon)}_{[\alpha,\beta)} \in \cS(\mathbb{R})$
  suitably converging to 
$1_{[\alpha,\beta)}$ as $\epsilon \to 0^+$. In this way,  the identity
\begin{equation}
		\label{PVMAB2}
		P^{(aA+bB)}_{[\alpha,\beta)}
		=\lim_{\epsilon \to 0^+} \lim_{N\to +\infty}\sum_{\lambda_j \in \sigma(A), \mu_j\in \sigma(B)} 1^{(\epsilon)}_{[\alpha,\beta)}\left( \frac{1}{N}\sum_{j=1}^N (a\lambda_j+ b\mu_j)\right) P^{(A)}_{\lambda_1} 
 P^{(B)}_{\mu_1}\cdots P^{(A)}_{\lambda_N}  
P^{(B)}_{\mu_N}
	\end{equation}
 can be established. We shall prove it later in a much more general context as identity (\ref{PVMAB}).

The requirement $f, \check{f} \in L^2(\bR,dx)$  can be relaxed  making stronger the regularity requirement on $f$ as a consequence of the following argument.  If $f,g$ satisfy (i),(ii),(ii) and also  $f(x)=g(x)$ for $x \in [-\|aA\|-\|bB\|, \|aA\|+ \|bB\|]$, where the norm is the operator norm, then 
the well-known general estimate for selfadjoint  operators  that $\sigma(T) \subset [-\|T\|,\|T\|]$, yields both
 \begin{equation} \label{idfg}f(aA+bB)=
g(aA+bB)\end{equation} from (\ref{elemspectr}),  and $$f\left( \frac{1}{N}\sum_{j=1}^N (\lambda_j+ \mu_j)\right) =
 g\left( \frac{1}{N}\sum_{j=1}^N (\lambda_j+ \mu_j)\right)\quad \mbox{for}\quad \lambda_j \in \sigma(A), \mu_j\in \sigma(B)\:.$$
 It is consequently safe to extend the validity of (\ref{Eq: f(aA+bB) in the finite dimensional case})   to every complex valued $f\in C^\infty(\bR)$ simply smoothly changing $f$ to a function of $\cS(\bR)$ 
 outside the interval  $[-\|aA\|-\|bB\|, \|aA\|+ \|bB\|]$. In particular, $f$ can be chosen as a polynomial of a single real variable.
In this way, we also obtain the elementary  {\em Jordan-algebra operations}  on $A$ and $B$ written in terms of the PVMs of $A$ and $B$ as the second identity below. If $a,b \in \bR$,
\begin{equation}
aA+bB =  
\lim_{N\to +\infty}\sum_{\lambda_j \in \sigma(A), \mu_j\in \sigma(B)}\left( a\sum_{j=1}^N \frac{\lambda_j}{N} 
+ b \sum_{j=1}^N \frac{\mu_j}{N}\right)\:\:  P^{(A)}_{\lambda_1} 
 P^{(B)}_{\mu_1}\cdots P^{(A)}_{\lambda_N}  
P^{(B)}_{\mu_N}\end{equation}
and
\begin{equation}
\frac{1}{2}(AB+BA) =  
\lim_{N\to +\infty}\sum_{\lambda_j \in \sigma(A), \mu_j\in \sigma(B)}\left(\sum_{j=1}^N \frac{\lambda_j}{N}\right)
\left(\sum_{j=1}^N \frac{\mu_j}{N}\right) \:\:  P^{(A)}_{\lambda_1}
 P^{(B)}_{\mu_1}\cdots P^{(A)}_{\lambda_N}  
P^{(B)}_{\mu_N}\:.\label{jordan2}\end{equation}
 The former is simply obtained from (\ref{Eq: f(aA+bB) in the finite dimensional case}) choosing $f(s)=s$. The latter
 arises from  a more involved though elementary procedure. The left-hand side of (\ref{jordan2})
can be written as $$\frac{1}{2}[(A+B)^2 -A^2-B^2] = f(aA+bB)+g(a'A+b'B) +g(a''A+b''B)$$
 where  $f(s):=s^2$, $g(s)= -s^2$, $a=b=1$, $a'=1$, $b'=0$, $a''=0$, $b''=1$.
 Linearity in $f$ of the right-hand side of (\ref{Eq: f(aA+bB) in the finite dimensional case}) 
(also with different constants $a,b$)  immediately gives rise to (\ref{jordan2}).
A more detailed discussion will appear in section \ref{Sec: some applications}, where (\ref{jordan2}) is extended to polynomials (Theorem \ref{teopolynomials} ).

With an argument similar to the one leading to (\ref{idfg}) and  using the fact that $\sigma(aA+bB)$ is  a finite discrete set
of reals, we achieve another useful result.  If $\beta \not \in \sigma(aA+bB)$, identity (\ref{PVMAB2}) can be alternatively written
\begin{equation}
		\label{PVMAB3}
		P^{(aA+bB)}_{[\alpha,\beta)}
		=\lim_{N\to +\infty}\sum_{\substack{\lambda_j \in \sigma(A)\\\mu_j\in \sigma(B)}} 1'_{[\alpha,\beta)}\left( \frac{1}{N}\sum_{j=1}^N (a\lambda_j+ b\mu_j)\right) P^{(A)}_{\lambda_1} 
 P^{(B)}_{\mu_1}\cdots P^{(A)}_{\lambda_N}  
P^{(B)}_{\mu_N}
	\end{equation}
where $1'_{[\alpha,\beta)} \in \cD(\bR)$ is a map which attains the constant value $1$ over $[\alpha, \beta)$ and smoothly 
vanishes on $\sigma(aA+bB) \setminus [\alpha,\beta)$.

When trying to extend (\ref{Eq: f(aA+bB) in the finite dimensional case}) to the case of a general (separable) Hilbert space $\sH$
 for generally unbounded selfadjoint operators $A$ and $B$, several evident technical issues arise. First of all,
 usual domain problems have to be fixed. These problems are tantamount  to corresponding domain  issues
with Trotter's formula which are well known and definitely  fixed \cite[Theorem VIII.31]{RS1}. A sufficiently general  setup consists 
of  assuming that $aA+bB$
 is essentially selfadjoint over its natural domain $D(A)\cap D(B)$ (if $ab\neq 0$). A much harder problem is the interpretation of the operator-valued
 integral appearing in the right-hand side of \eqref{Eq: f(aA+bB) in the finite dimensional case}, especially in case the spectrum
 of either $A$ or $B$ includes a continuous part.  Actually,  problems of a similar nature arise also when $A$ and $B$ have pure 
 point spectra, but $\sH$ is infinite dimensional.
In the following sections we will address those technical problems and other related ones, ending up with a wide generalization 
of \eqref{Eq: f(aA+bB) in the finite dimensional case} -- \textit{cf.} Theorem \ref{teorapgen}.

\subsection{The unital Banach algebra of complex Borel measures}

Before entering the technical details of the core of the paper 
it is necessary to introduce the  space of functions we will use on both sides of the extension of 
(\ref{Eq: f(aA+bB) in the finite dimensional case}).

\begin{definition} {\em The {\bf space of complex measures on $\bR$} denoted by $\cM(\bR)$ is the complex linear space of 
$\mathbb{C}$-valued $\sigma$-additive maps $\nu : \cB(\bR) \to \bC$.  The  space of {\bf Fourier transforms of complex measures on $\bR$}  consists of 
the complex linear space $\mathscr{F}(\mathbb{R})$ of  functions of the form
\begin{equation}\label{f-muf}f_\nu(x):= \int_\bR e^{ixy}d\nu(y), \qquad x\in \bR,\end{equation}
for some  $\nu\in \cM(\bR)$. \hfill $\blacksquare$ }
\end{definition}
\noindent The linear map
\begin{align}\label{Eqn: Fourier transform for measure}
	F\colon \cM(\bR) \ni \nu \mapsto f_\nu \in \cF(\bR)\,,
\end{align}
is injective (this is a straightforward extension of \cite[Theorem 26.2]{Bill}) so that it defines a linear isomorphism. From the standard 
properties of Fourier transform, it is not difficult to prove that  $\cF$ includes  the maps satisfying (i),(ii),(iii) we used  in section \ref{finite}. Furthermore
$$\cF(\bR) \subset C_b(\bR) \:,\quad  \cS(\bR) \subset \cF(\bR) \subset \cS'(\bR)\:, \quad \cD(\bR) \subset \cF(\bR)\subset \cD'(\bR)\:.$$
Above, $C_b(\bR)$ is the commutative unital Banach algebra  (with norm $\|\cdot \|_\infty$) of bounded continuous complex-valued functions over 
$\bR$,  $\cS(\bR)$ and $\cD(\bR)$
are respectively the {\em space of Schwartz (complex) functions} and the {\em space of the (complex) test functions} over $\bR$, and 
 $\cS'(\bR)$ and $\cD'(\bR)$
denote the  associated  spaces of distributions.\\
 The linear space $\cM(\bR)$ turns out to be a {\em commutative unital  Banach algebra}  where 
(a) the product of two measures $\nu$ and $\nu'$ is  their {\em convolution} $\nu*  \nu'$, (b)
 the unit is the Dirac point mass $\delta _0$ concentrated at $0$, (c)
 the norm of a measure $\nu$ is defined as its total variation $\|\nu\|=|\nu|(\bR)$. $\cM(\bR)$ also admits a norm-preserving  unit-preserving
  antilinear {\em involution} given by the complex conjugation of measure $\nu^* := \overline{\nu}$.
The linear isomorphism  (\ref{Eqn: Fourier transform for measure}) induces   a unital commutative  Banach algebra structure over $\cF(\bR)$ when 
defining  $\|f_\nu\|_\cF:=\|\nu\|$. More precisely, the map
$F$ promoted to an  isomorphism of Banach algebras transforms the convolution of measures into the pointwise product of corresponding 
functions and the unit of $\cM(\bR)$ to the constant function $1$.  The involution of $\cM(\bR)$ becomes the norm-preserving  unit-preserving  
antilinear involution over $\cF(\mathbb{R})$, given by $f^*(x) := \overline{f(-x)}$,  the bar denoting the complex conjugation.
The definition  of $\|f\|_\cF$ easily implies 
\begin{equation}\|f\|_\infty \leq \|f\|_\cF\:, \quad f \in \cF(\bR)\:.\end{equation}
If $f \in \cS(\bR)$ (or $\cD(\bR)$), then $f= f_\nu$ where $d\nu = (2\pi)^{-1} \check{f} dx$ and $\check{f}(y) = \int_\bR e^{-ixy} f(x) dy$ is
 the {\em Fourier anti-transform} of $f$. In this case
 \begin{equation}\label{normS} \|f\|_\cF =  (2\pi)^{-1}  \|\check{f}\|_{L^1(\bR,dx)}\:, \quad f \in \cS(\bR) \:.\end{equation}

\subsection{Regularized products of  non-commuting PVMs of   selfadjoint operators}

This section and the subsequent one are devoted to extend formula (\ref{Eq: f(aA+bB) in the finite dimensional case}) to the case of an infinite-dimensional,
though separable, Hilbert space $\sH$. We will suppose that the operators $A$ and $B$ are  
unbounded selfadjoint operators with domains $D(A)$ and $D(B)$ respectively, and that their  linear combination $aA+bB$, for suitable
 $a,b \in \bR$, is essentially selfadjoint on its {\em standard 
domain}  according to section \ref{Sec: introduction}.

 In other words, we want  to construct a functional calculus and the spectral  measure of the selfadjoint operator given by the  closure $\overline{aA+bB}$ 
out of the PVMs $P^{(A)}, P^{(B)}$ of $A$ and $B$ respectively, 
 proving the following suggestive formula if $\psi \in \sH$ and  $f \in \cF(\bR)$:
\begin{equation}\label{mainid}
f\left(\overline{aA+bB}\right) \psi = \lim_{N\to +\infty} \int_{\bR^{2N}} f\left( \frac{1}{N}\sum^N_{n=1} (a\lambda_{n} + b\mu_{n})\right)
 dP^{(A)}_{\lambda_1}dP^{(B)}_{\mu_1}\cdots dP^{(A)}_{\lambda_N}dP^{(B)}_{\mu_N}\psi\:.
\end{equation}
 The operator $f\left(\overline{aA+bB}\right)$ on the left-hand side is independently  defined by (\ref{Eq: function of self-adjoint operator})
 and its domain is the whole $\sH$
according to (\ref{Eq: function of self-adjoint operator}) since $f$ is bounded.

As a first step, we address the problem of the  interpretation of the operators 
\begin{equation}\label{mainid-1}\int_{\bR^{2N}} f\left( \frac{1}{N}\sum^N_{n=1} (a\lambda_{n} + b\mu_{n})\right) dP^{(A)}_{\lambda_1}dP^{(B)}_{\mu_1}
\cdots dP^{(A)}_{\lambda_N}dP^{(B)}_{\mu_N}\end{equation} after the symbol of  limit on the right-hand side of \eqref{mainid}. In fact,  
 we cannot generally interpret the integration appearing in \eqref{mainid-1} like the one in \eqref{prototypeC}, i.e. as  referred to the
 {\em joint PVM}  of $N$ copies of $P^{(A)}$ and $P^{(B)}$ 
(see, e.g., \cite{RS2,Moretti2017}). This is because  we are focusing on the generic situation where $P^{(A)}$ and $P^{(B)}$ do {\em not} commute.
In the general case, the operators \eqref{mainid-1} are  defined out of a regularized natural quadratic form we are going to
 introduce with the following crucial technical result.  We stress that we will not use the hypothesis of essential selfadjointness 
of $aA+bB$ at this stage of the construction.

\begin{lemma}\label{proptec} Let $A$ and $B$   be a pair of selfadjoint operators over the separable Hilbert space $\sH$ 
as in (\ref{Eq: function of self-adjoint operator}), and  $Q:=\{Q_n\}_{n\in \bN}$ a sequence of orthogonal projectors over respective
 finite-dimensional subspaces such that $Q_n \to I$ strongly as $n\to +\infty$.\\
 If $\phi, \psi \in \sH$ and $N\in\mathbb{N}$, then  the following facts are valid.
\begin{itemize}
\item[{\bf (a)}]  For every $n\in \bN$,  there is a unique complex Borel measure $\nu_{\phi,\psi,Q}^{(N,n)}$ over $\bR^{2N}$ such that
\begin{equation}\label{reg-measures}
	\nu^{(N,n)}_{\phi,\psi,Q}(\times_{\ell=1}^NI_\ell\times J_\ell):=
	\langle\phi | P^{(A)}_{I_1}Q_{n}P^{(B)}_{J_1}\cdots P^{(A)}_{I_N}Q_{n}P^{(B)}_{J_N}\psi\rangle,
	\end{equation}
	with $I_\ell, J_\ell$ arbitrary Borel sets in $\bR$. The support of the measure satisfies
\begin{equation}\label{supp-reg-measures}\mbox{supp}(|\nu^{(N,n)}_{\phi,\psi,Q}|) \subset 
	\underbrace{\sigma(A)\times \sigma(B) \times \cdots \times \sigma(A) \times \sigma(B)}_{2N \ \mathrm{ times}}\:.
	\end{equation}

\item[{\bf (b)}]  For every $f \in \cF(\bR)$  and, independently of the choice of the sequence $\{Q_n\}_{n\in \bN}$, the following
 limit holds  for $a,b \in \bR$:
\begin{multline}\label{limnuNn}
	\lim_{n\to \infty }
	\int _{\bR^{2N}} f\left(\frac{1}{N}\sum_{l=1}^{N}(a\lambda _l+b\mu_l)\right)d\nu^{(N,n)}_{\phi,\psi,Q}(\lambda_1,\mu_1,\dots,\lambda _N,\mu_N)\\
	=\int_\bR\left\langle \phi \left|  (e^{i \frac{t}{N}aA}e^{i \frac{t}{N}bB})^N\psi\right.\right\rangle d\nu _f(t)\,,
 \end{multline}
with $\nu_f := F^{-1}(f)$ according to (\ref{f-muf}-\ref{Eqn: Fourier transform for measure}).
\end{itemize}
\end{lemma}

\begin{proof} See Appendix \ref{App: proof of some propositions}.
\end{proof}

\noindent Analogously to  the  use of (\ref{spectral2}) to define the right-hand side in (\ref{Eq: function of self-adjoint operator}),
 the idea is now to define the operator in (\ref{mainid-1})  in order that  the associated quadratic form 
coincides to the  limit  of the regularized sequences as on the left-hand side of (\ref{limnuNn}). 

\begin{theorem}\label{defint}  Let $A$ and $B$   be a pair of selfadjoint operators over the separable
 Hilbert space $\sH$ as in (\ref{Eq: function of self-adjoint operator}), $f\in \cF$, and $N\in\mathbb{N}$. 
For $a,b \in \bR$, there is a unique  operator, denoted by
\begin{gather}\label{Eq: operator integral w.r.t. product of 2N-projectors}
	\int_{\bR^{2N}} f\left( \frac{1}{N}\sum^N_{n=1}  (a\lambda_{n} +b \mu_{n})\right)
	dP^{(A)}_{\lambda_1}dP^{(B)}_{\mu_1}\cdots dP^{(A)}_{\lambda_N}dP^{(B)}_{\mu_N} \in \gB(\sH)\:,
\end{gather}
such that 
\begin{multline}
\left\langle \phi \left| \int_{\bR^{2N}} f\left( \frac{1}{N}\sum^N_{n=1}  (a\lambda_{n} + b\mu_{n})\right) dP^{(A)}_{\lambda_1}dP^{(B)}_{\mu_1}\cdots dP^{(A)}_{\lambda_N}dP^{(B)}_{\mu_N}\psi \right. \right\rangle \\
=\lim_{n\to \infty } \int _{\bR^{2N}} f\left(\frac{1}{N}\sum_{l=1}^{N}(a\lambda _l+b\mu_l)\right)d\nu^{(N,n)}_{\phi,\psi,Q}(\lambda_1, \mu_1, \dots , \lambda _N,\mu_N),
 \quad \mbox{if $\phi, \psi \in \sH$.}\label{DEFint}
\end{multline}
Above,  for every   $n\in \bN$, the complex Borel  measure $\nu^{(N,n)}_{\phi,\psi,Q}$ is uniquely defined by
$$\nu^{(N,n)}_{\phi,\psi,Q}(\times_{\ell=1}^NI_\ell\times J_\ell):=
	\langle\phi | P^{(A)}_{I_1}Q_{n}P^{(B)}_{J_1}\cdots P^{(A)}_{I_N}Q_{n}P^{(B)}_{J_N}\psi\rangle\quad \mbox{with $I_l, J_l \in \cB(\bR)$}$$
 for a sequence of finite-dimensional orthogonal projectors $Q:=\{Q_n\}_{n\in \bN}$ with $Q_n \to I$ strongly as $n\to +\infty$. The limit in (\ref{DEFint})
is however independent of the choice of  $\{Q_n\}_{n\in \bN}$.\\
The following further facts are true. 
\begin{itemize}
\item[(i)]  The map $$\cF \ni f \mapsto \int_{\bR^{2N}} f\left( \frac{1}{N}
\sum^N_{n=1}  (a\lambda_{n} + b\mu_{n})\right) dP^{(A)}_{\lambda_1}dP^{(B)}_{\mu_1}\cdots dP^{(A)}_{\lambda_N}dP^{(B)}_{\mu_N}\in\mathfrak{B}(\mathsf{H})$$ is linear.
\item[(ii)] The following  inequality holds
\begin{equation}
\left|\left|   \int_{\bR^{2N}} f\left( \frac{1}{N}\sum^N_{n=1}  (a\lambda_{n} + b\mu_{n})\right) dP^{(A)}_{\lambda_1}dP^{(B)}_{\mu_1}
\cdots dP^{(A)}_{\lambda_N}dP^{(B)}_{\mu_N} \right|\right| \leq \|f\|_\cF\label{ineq2}\:.
\end{equation}
\item[(iii)] If $\phi,\psi \in \sH$, the restricted maps
\begin{align}
 &\cS(\bR)\ni f \mapsto \left\langle \phi \left| \int_{\bR^{2N}} f\left( \frac{1}{N}\sum^N_{n=1}  (a\lambda_{n} + b\mu_{n})\right) dP^{(A)}_{\lambda_1}dP^{(B)}_{\mu_1}\cdots dP^{(A)}_{\lambda_N}dP^{(B)}_{\mu_N}\psi \right. \right\rangle \label{dist1}\\
 &\cS(\bR)\ni f \mapsto \int_{\bR^{2N}} f\left( \frac{1}{N}\sum^N_{n=1}  (a\lambda_{n} + b\mu_{n})\right) dP^{(A)}_{\lambda_1}dP^{(B)}_{\mu_1}\cdots dP^{(A)}_{\lambda_N}dP^{(B)}_{\mu_N} \label{dist2}
\end{align}
are, respectively, a Schwartz distribution and a $\cB(\sH)$-valued Schwartz distribution.
The same facts hold if replacing $\cS(\bR)$ with $\cD(\bR)$ and $\cS'(\bR)$ with $\cD'(\bR)$.
\item[(iv)] If $A, B \in \gB(\sH)$ and 
$\mbox{supp}(f) \cap [-\|aA\| -\|bB\|, \|aA\| + \|bB\|] = \emptyset$, then
$$\int_{\bR^{2N}} f\left( 
\frac{1}{N}\sum^N_{n=1} (a\lambda_{n} + b\mu_{n})\right) dP^{(A)}_{\lambda_1}dP^{(B)}_{\mu_1}\cdots dP^{(A)}_{\lambda_N}dP^{(B)}_{\mu_N}=0\:.$$
\item[(v)]  If $\overline{z}$ denotes the complex conjugate of $z\in \mathbb{C}$, then
\begin{multline}
	\left[\int_{\bR^{2N}} f\left( 
	\frac{1}{N}\sum^N_{n=1} (a\lambda_{n} + b\mu_{n})\right) dP^{(A)}_{\lambda_1}dP^{(B)}_{\mu_1}\cdots dP^{(A)}_{\lambda_N}dP^{(B)}_{\mu_N}\right]^*\\
	=\int_{\bR^{2N}} \overline{f\left( 
	\frac{1}{N}\sum^N_{n=1} (a\lambda_{n} + b\mu_{n})\right)} dP^{(B)}_{\mu_1}dP^{(A)}_{\lambda_1}\cdots dP^{(B)}_{\mu_N}dP^{(A)}_{\lambda_N}\,.
\end{multline}
\end{itemize}
\end{theorem}

\begin{proof} If $\nu_f := F^{-1}(f)$ according to (\ref{f-muf})-(\ref{Eqn: Fourier transform for measure}), the map appearing on the right-hand side of (\ref{limnuNn}), that is
$$\sH\times \sH \times \cF \ni (\phi, \psi, f) \mapsto I^{(N)}_{\phi,\psi}(f) := \int_\bR\left\langle \phi \left|  (e^{i \frac{t}{N}aA}e^{i \frac{t}{N}bB})^N\psi\right.\right\rangle d\nu _f(t)$$
is by construction linear in $f$ and $\psi$, antilinear in $\phi$ and satisfies 
\begin{equation} |I^{(N)}_{\phi,\psi}(f)| \leq \|\phi\|_\sH \|\psi\|_\sH \|f\|_\cF \label{ineq}\:.\end{equation}
As a consequence, identity 
(\ref{limnuNn}) permits us to exploit  Riesz' lemma defining the wanted operator as required and independently
 from the regularizing sequence $\{Q_n\}_{n\in \bN}$. Indeed, for any $\psi,f$, the map $\phi\mapsto I^{(N)}_{\phi,\psi}(f)$ defines a bounded antilinear functional, hence Riesz representation theorem allows to write the following identity
 $$  I^{(N)}_{\phi,\psi}(f)=\langle \phi, \xi_f\rangle $$
 for a suitable vector $ \xi_f\in \sH$ such that \begin{equation}\label{bound-1} \|\xi_f\|\leq  \|\psi\|_\sH \|f\|_\cF.\end{equation} The linear dependence of  $I^{(N)}_{\phi,\psi}(f)$ on $\psi$ yields the linearity of   the map $\psi \mapsto \xi_\psi:=L(\psi)$ as well. Finally, inequality \eqref{bound-1} allows to prove that the operator $L$ is bounded and by definition it coincides with \eqref{Eq: operator integral w.r.t. product of 2N-projectors}.
Linearity of $f \mapsto I^{(N)}_{\phi,\psi}(f)$ and (\ref{ineq})  prove (i) and (ii). \\
Regarding (iii), we observe that   from the continuity properties of the Fourier transform,
 $\cS(\bR) \ni f_n \to 0$ in  the $\cS$-topology  (as $n \to +\infty$) implies that the sequence of anti-transforms $\check{f}_n$ vanishes
 in the same topology, so that, in particular,   $\|\check{f}_n\|_{L^1(\bR, dx)} \to 0$.  Since $\|f_n\|_\cF = (2\pi)^{-1}  \|\check{f}_n\|_{L^1(\bR, dx)}$ from (\ref{normS}),  we have that
$\|f_n\|_\cF \to 0$.  The linear maps (\ref{dist1})  and (\ref{dist2}) are therefore 
$\cS$-continuous in view of inequality (\ref{ineq2}) so that they are Schwartz distributions by definition. Since Schwartz
 distributions are distributions over $\cD(\bR)$,
 the result is also valid when restricting our map to $\cD(\bR)$.\\(iv) First observe that   $\sigma(A) \subset [-\|A\|,\|A\|]$ 
and $\sigma(B) \subset [-\|B\|, \|B\|]$. At this juncture, 
 inclusion (\ref{supp-reg-measures}) and the definition (\ref{DEFint}) imply that the integration over $\bR^{2N}$  on the right-hand side of  (\ref{DEFint})
can actually be restricted to the set determined by imposing 
$\lambda_j \in [-\|A\|,\|A\|]$ and $\mu_j\in [-\|B\|,|B\|]$ for $j=1,\ldots, N$, 
without affecting the final result.  If the support of $f$ is disjoint from  $ [-\|aA\|-\|bB\|, \|aA\|+\|bB\|]$, then  
$f\left(\frac{1}{N}\sum_{j=1}^N (a\lambda_j+b\mu_j)\right)$ vanishes 
when $\lambda_j \in [-\|A\|,\|A\|]$ and $\mu_j\in [-\|B\|,\|B\|]$ 
so that  the integral in  (\ref{DEFint}) vanishes as well for every $n$,  producing $0$ as limit for $n\to +\infty$.  \\ 
The proof of (v) follows at once from the definitions.
\end{proof}

\begin{remark}{\em $\null$\\
{\bf(1)} If $\sH$ is finite-dimensional, all $Q_n$ on the right-hand side of  (\ref{reg-measures}) can be removed,  simply choosing 
$Q_n=I$, finding the formulas achieved in section \ref{finite}.
For  infinite-dimensional $\sH$,  there are cases where no complex measure satisfying the identity (\ref{reg-measures})  exists
if $Q_n=I$ for all $n$ -- \textit{cf.} \ref{Sec: The Q_n-regularization  is generally necessary}.\\
{\bf (2)} The existence of  sequences $\{Q_n\}_{n\in \bN}$ strongly approximating the identity operator with {\em finite-dimensional}
 projection  spaces  is equivalent to  {\em separability} of $\sH$. Hence, separability hypothesis  cannot be relaxed.
}\hfill 
$\blacksquare$ 
\end{remark}

\subsection{Main result}
We have reached  a position to establish  our main theorem when, in addition to the already assumed  hypotheses on $A$ and $B$,
 we use the requirement that $aA+bB$
is essentially selfadjoint.

\begin{theorem}\label{teorapgen}
	Let $A,B$ be self-adjoint operators on a separable   Hilbert space, $a,b\in \bR$ and let us assume that 
  $aA+bB$ is essentially self-adjoint on $D(aA)\cap D(bB)$. Then the following facts hold for $f\in \cF(\bR)$.
\begin{itemize}
\item[{\bf (a)}]
Denote the strong operator limit by   $s-\lim$. Then
\begin{equation}\label{mainid2}
f\left(\overline{aA+bB}\right)  = \mbox{s-}\lim_{N\to +\infty} \int_{\bR^{2N}} f\left( \frac{1}{N}\sum^N_{n=1}  
(a\lambda_{n} + b\mu_{n})\right) dP^{(A)}_{\lambda_1}dP^{(B)}_{\mu_1}\cdots dP^{(A)}_{\lambda_N}dP^{(B)}_{\mu_N}\:,
\end{equation}
and a similar identity is valid if, keeping the left-hand side, we swap $dP^{(A)}_{\lambda_j}$ and $dP^{(B)}_{\mu_j}$ for $j=1,2,\ldots, N$ on the right-hand side.
\item[{\bf (b)}] If  the PVMs $P^{(aA)}$ and $P^{(bB)}$ commute (this fact holds in particular if either $a=0$ or $b=0$), then
\begin{equation}\label{mainid22}
f\left(\overline{aA+bB}\right)  =  \int_{\bR^{2N}} f\left( \frac{1}{N}\sum^N_{n=1}  (a\lambda_{n} + b\mu_{n})\right) dP^{(A)}_{\lambda_1}dP^{(B)}_{\mu_1}\cdots dP^{(A)}_{\lambda_N}dP^{(B)}_{\mu_N}\:,
\end{equation}
for every  $N\in\mathbb{N}$.  

\item[{\bf (c)}] If the PVMs $P^{(A)}$ and $P^{(B)}$ commute, then 
$$\nu^{(1,n)}_{\phi,\psi,Q}(I\times J)=
	\langle\phi | P^{(A)}_{I}Q_{n}P^{(B)}_{J}\psi\rangle \to \langle\phi| P^{(A,B)}_{I\times J}\psi \rangle \quad \mbox{as $n\to +\infty$}
$$ where   $P^{(A,B)} : \cB(\bR^2) \to \cL(\sH)$ is the joint PVM of $A$ and $B$ and the right-hand side of (\ref{mainid22}) -- defined as 
in (\ref{DEFint}) -- with $N=1$ 
can be interpreted as an integration with respect to that joint PVM.
\end{itemize}
\end{theorem}

\begin{proof} 
(a) From now on, $\nu_f := F^{-1}(f)$ according to (\ref{f-muf})-(\ref{Eqn: Fourier transform for measure}).   For $N\in \bN$, we define the operator
$$\int_\bR \left(e^{i \frac{t}{N}aA}e^{i \frac{t}{N}bB}\right)^N d\nu _f(t)$$
as the unique operator  $S\in \gB(\sH)$ such that
$$\langle \phi| S\psi\rangle = \int_\bR\left\langle \phi \left|  (e^{i \frac{t}{N}aA}e^{i \frac{t}{N}bB})^N\psi\right.\right\rangle d\nu _f(t)\quad \mbox{if $\phi,\psi\in \sH$}\:.$$
The operator
\begin{equation}\int_\bR e^{i t(\overline{aA+bB})}d\nu _f(t)\label{ininin}\end{equation}
is analogously defined. \\
The former operator is already known. Indeed,  from (\ref{limnuNn}) and  (\ref{DEFint}), we find
\begin{equation}\int_\bR \left(e^{i \frac{t}{N}aA}e^{i \frac{t}{N}bB}\right)^N d\nu _f(t)= \int_{\bR^{2N}} f\left( \frac{1}{N}\sum^N_{n=1}  
(a\lambda_{n} + b\mu_{n})\right) dP^{(A)}_{\lambda_1}dP^{(B)}_{\mu_1}\cdots dP^{(A)}_{\lambda_N}dP^{(B)}_{\mu_N}\:.\label{ON}\end{equation}
We now turn our attention to the operator in (\ref{ininin}). Let $P^{(\overline{aA+bB})}$ be the spectral measure of the self-adjoint operator $\overline{aA+bB}$ 
and  set 
$\mu_{\phi, \psi}(E) = \langle \phi|P^{(\overline{aA+bB})}(E)\psi \rangle$, $E \in \cB(\bR)$, the complex Borel measure associated to  $\phi$ and $\psi$.
By standard functional calculus, by writing $f\in \cF $ as in \eqref{f-muf} 
and exploiting Fubini's theorem (all measures are finite and the integrand is bounded), we obtain:
\begin{align*}
	\langle \phi| f(\overline{aA+bB}) \psi \rangle&=\int_\bR f(\lambda)d\mu_{\phi,\psi}(\lambda)
	=\int_\bR \int_\bR e^{i\lambda t}d\nu_f (t) d\mu_{\phi,\psi}(\lambda)
	=\int_\bR \int_\bR e^{i\lambda t}d\mu_{\phi,\psi}(\lambda) d\nu _f(t)\\
	&=\int_\bR \langle \phi| e^{i t(\overline{A+B})}\psi\rangle d\nu_f (t).
\end{align*}
In other words,
\begin{equation} f(\overline{aA+bB}) = \int_\bR  e^{i t(\overline{aA+bB})} d\nu_f (t)\label{bbbb}
\end{equation}
Taking advantage of (\ref{ON}) and (\ref{bbbb}), the thesis of the theorem can be rephrased to
\begin{equation}\left|\left| \int_\bR e^{i t(\overline{aA+bB})} d\nu_f (t) \psi  - \int_\bR \left(e^{i \frac{t}{N}aA}e^{i \frac{t}{N}bB}\right)^N d\nu_f (t) \psi\right|\right|^2 \to 0 \quad \mbox{for $N\to +\infty$ if $\psi \in \sH$.}\label{thes2}\end{equation}
Representing the left-hand side as an inner product and expanding it, taking the previous definitions into account, we find 
that the squared norm above equals
$$\int_\bR \left( \int_\bR \left\langle  \left. e^{i t(\overline{aA+bB})} \psi - \left(e^{i \frac{t}{N}aA}e^{i \frac{t}{N}bB}\right)^N\psi   \right|          
e^{i s(\overline{aA+bB})} \psi - \left(e^{i \frac{s}{N}aA}e^{i \frac{s}{N}bB}\right)^N \psi  \right\rangle d\overline{\nu_f}(t)\right) d\nu_f(s)\:.$$
Fubini's theorem permits us to re-write the above integral as
$$\int_{\bR^2} \left\langle  \left. e^{i t(\overline{aA+bB})}\psi  - \left(e^{i \frac{t}{N}aA}e^{i \frac{t}{N}bB}\right)^N\psi   \right|          
e^{i s(\overline{aA+bB})} \psi - \left(e^{i \frac{s}{N}aA}e^{i \frac{s}{N}bB}\right)^N\psi  \right\rangle d\overline{\nu_f}(t) \otimes d\nu_f(s)\:,$$
where  $d\overline{\nu_f}(t) \otimes d\nu_f(s)$ is the product measure.
By the {\em Trotter product formula}  \cite[Theorem VIII.31]{RS1}, which is valid in our hypotheses on $aA$ and $bB$,  both 
entries of the scalar product vanish as $N\to +\infty$. This fact implies that the integral itself vanishes as a consequence of Lebesgue's 
dominated convergence theorem, since the product measure is finite and the integrand is uniformly bounded in $(t,s)$ as the involved operators are unitary.
 In summary,
(\ref{thes2}) is valid and the proof of (a) ends since the last statement is an obvious consequence of $\overline{aA+bB}=\overline{bB+aA}$.

(b) Exploiting the structure of the proof of (a), the thesis is valid if, for every $t\in \bR$,
$$e^{i t(\overline{aA+bB})} = \left(e^{i \frac{t}{N}aA}e^{i \frac{t}{N}bB}\right)^N\:.$$
Let us prove that this identity is in fact true. Since $P^{(aA)}$ and $P^{(bB)}$ commute, referring to their joint PVM $P$, we have
$$e^{i\frac{t}{N}aA}= \int_{\bR^2} e^{i\frac{t}{N}\lambda}dP(\lambda,\mu)\:, \quad e^{i\frac{t}{N}bB}= \int_{\bR^2} e^{i\frac{t}{N}\mu}dP(\lambda,\mu)\:.$$
From the general properties of the integral of bounded functions  with respect to a given PVM,
$$\left(e^{i \frac{t}{N}aA}e^{i \frac{t}{N}bB}\right)^N = 
\int_{\bR^2}\left(e^{i \frac{t}{N}\lambda}e^{i \frac{t}{N}\mu}\right)^NdP(\lambda,\mu) 
= \int_{\bR^2} e^{i t \lambda}e^{i t \mu}dP(\lambda,\mu) = e^{i taA}e^{i tbB}\:.$$
Taking the strong limit as $N\to +\infty$, Trotter's formula yields
$$e^{i t(\overline{aA+bB})} =e^{i taA}e^{i tbB}\:,$$
so that, using again 
$\left(e^{i \frac{t}{N}aA}e^{i \frac{t}{N}bB}\right)^N =  e^{i taA}e^{i tbB}$,
we find
$$e^{i t(\overline{aA+bB})} = \left(e^{i \frac{t}{N}aA}e^{i \frac{t}{N}bB}\right)^N\:,$$
as wanted. 

(c)  is evident from (b),  the requirement that $Q_n \to I$ strongly, the definition of joint PVM, and its elementary properties.
\end{proof}

\section{Some applications}\label{Sec: some applications}

\subsection{The PVM of $\overline{aA+bB}$ from $P^{(A)}$ and $P^{(B)}$}\label{Rmk: PVM of aA+bB in terms of PVM of A,B}
As an application of Theorem \ref{teorapgen},  we prove  that knowledge of the PVMs $P^{(A)},P^{(B)}$ of $A$ and $B$ suffices to determine 
the PVM of $\overline{aA+bB}$ among all PVMs over $\mathbb{R}$.  
We take advantage of  the well-known result  that\footnote{From standard extension theorems of finite 
Borel measures as \cite[Corollary 1.6.2]{Cohn} (whose proof is also valid for complex measures as the reader proves easily), there is only one regular complex 
Borel measure $\mu$ over $\mathbb{R}$  determined by the values
 $\mu([\alpha, \beta))$. If two PVMs $P,Q$ satisfy $\langle \phi| Q_{[\alpha,\beta)}\psi \rangle = \mu_{\phi,\psi}([\alpha, \beta)) =\langle \phi| P_{[\alpha,\beta)}\psi \rangle$ for all $\phi,\psi \in \sH$ and $\alpha < \beta$ in $\mathbb{R}$, then it must also be
$\langle \phi| Q_{E}\psi \rangle = \mu^{(Q)}_{\phi,\psi}(E) =\mu^{(P)}_{\phi,\psi}(E) =\langle \phi| P_{E}\psi \rangle$ for every $E\in \cB(\mathbb{R})$ and $\phi,\psi \in \sH$, so that $P=Q$.}  $P^{(\overline{aA+bB})}$ is known once the projectors $P^{(\overline{aA+bB})}_{[\alpha,\beta)}$ are known for all $\alpha,\beta \in \bR$ with $\alpha < \beta$.

\begin{theorem}\label{teoPVM} If $A$ and $B$ are as in Theorem \ref{teorapgen} and $a,b \in \bR$, the following
 identity holds for every $\alpha,\beta \in \bR$ with $\alpha < \beta$,
\begin{equation}
		\label{PVMAB}
		P^{(aA+bB)}_{[\alpha,\beta)}
		=\mbox{s-}\lim_{\epsilon\to 0^+}\mbox{s-}\lim_{N\to +\infty}
		\int_{\bR^{2N}} 1^{(\epsilon)}_{[\alpha,\beta)}\left( \frac{1}{N}\sum^N_{n=1}  (a\lambda_{n} + b\mu_{n})\right)
		dP^{(A)}_{\lambda_1}dP^{(B)}_{\mu_1}\cdots dP^{(A)}_{\lambda_N}dP^{(B)}_{\mu_N}\,,
	\end{equation}
where the $\epsilon$-parametrized class of   $\cS(\mathbb{R})$ functions
\begin{equation}\label{chi}1^{(\epsilon)}_{[\alpha, \beta)}(x) :=
\exp\left\{-\frac{\epsilon}{(x-\beta)^2}  -\frac{\epsilon}{(x-\alpha)^2} \right\} 1_{(\alpha,\beta)}(x)+\exp\left\{-\frac{(x-\alpha)^2}{\epsilon}\right\}\:,\quad x \in \mathbb{R}
 \end{equation}
pointwise converges to $1_{[\alpha,\beta)}$ for $\epsilon \to 0^+$.
\end{theorem}

\begin{proof}  
	Observing that $\|1^{(\epsilon)}_{[\alpha,\beta)}\|_\infty\leq 2$ and 
$1^{(\epsilon)}_{[\alpha,\beta)} \to 1_{[a,b)}$ pointwise for $\epsilon \to 0^+$, the Lebesgue dominated convergence theorem
 and the finiteness of the positive spectral measure $\mu^{(\overline{aA+bB})}_{\psi,\psi}(E):=\langle\psi|P^{(\overline{aA+bB})}_E\psi\rangle$ 
imply that 
	\begin{align*}
		\int_{\mathbb R} |1^{(\epsilon)}_{[\alpha,\beta)}(\tau)- 1_{[\alpha, \beta)}(\tau)|^2d\mu^{(\overline{aA+bB})}_{\psi,\psi}(\tau)\stackrel{\epsilon\to 0^+}{\longrightarrow}
		0 \quad \mbox{for all $\psi\in\mathsf{H}$.}
	\end{align*}
Since $P^{(\overline{aA+bB})}_{[\alpha,\beta)}= 1_{[\alpha,\beta)}(\overline{aA+bB})$ and on account of the general identity  $\|f(\overline{aA+bB})\psi\|^2=\int_{\mathbb R}|f(\tau)|^2d\mu^{(\overline{aA+bB})}_{\psi,\psi}(\tau)$, we have
$P^{(\overline{aA+bB})}_{[\alpha,\beta)} = \mbox{s-}\lim_{\epsilon \to 0^+}1^{(\epsilon)}_{[\alpha,\beta)}(\overline{aA+bB})$.
Moreover, $1^{(\epsilon)}_{[\alpha,\beta)}\in \cS(\mathbb{\bR})$, so that it belongs to $\cF(\mathbb{R})$ and we can apply  Theorem \ref{teorapgen}
 to the right-hand side obtaining (\ref{PVMAB}).
\end{proof}

\begin{remark}\label{Rmk: FT of A+B-spectral measure in terms of FT of A,B-spectral measure}
	$\null$\\
{\em {\bf (1)} A statement analogous to that in 
Theorem \ref{teoPVM}  can be proved for the elements $P^{(\overline{aA+bB})}_{(\gamma,\alpha]}$ with $\gamma< \alpha$. Knowledge of 
$P^{(\overline{aA+bB})}_{(\gamma,\alpha]}$ and $P^{(\overline{aA+bB})}_{[\alpha, \beta)}$ permits one to construct the atom 
$P^{(\overline{aA+bB})}_{\{\alpha\}}$ since $P^{(\overline{aA+bB})}_{(\gamma,\alpha]} P^{(\overline{aA+bB})}_{[\alpha, \beta)} 
= P^{(\overline{aA+bB})}_{(\gamma,\alpha] \cap [\alpha, \beta)} = P^{(\overline{aA+bB})}_{\{\alpha\}}\:.$
Alternatively, $P^{(\overline{aA+bB})}_{\{\alpha\}}$ directly arises from
 (\ref{PVMAB}) when using $1^{(\epsilon)}_{\{\alpha\}}(x) := \exp{\{-(x-\alpha)^2/\epsilon\}}$
 in place of $1^{(\epsilon)}_{[\alpha, \beta)}(x) $.\\
{\bf (2)} If $\alpha$ and $\beta$ admit corresponding neighbourhoods  without  intersection with the difference 
$\sigma(aA+bB) \setminus [\alpha, \beta)$, 
 the $\epsilon$ limit can be removed from  the left-hand side of (\ref{PVMAB}) 
if replacing $1^{(\epsilon)}_{\{\alpha\}}$ with a map $1'_{[\alpha, \beta)} \in \cD(\mathbb{R})$
which constantly attains 
the value $1$ over $[\alpha, \beta)$ and rapidly vanishes on $\sigma(aA+bB) \setminus [\alpha, \beta)$.\\ 
	{\bf (3)} Theorem \ref{teoPVM} proves  that the PVM of $aA+bB$ can be determined from the PVMs of the operators $A$ and $B$.
	With a slightly different approach (assuming $a=b=1$ for the sake of simplicity) it is also possible to prove a more direct
 relation between the Fourier transform of the spectral measures $\mu^{(\overline{A+B})}_{\psi,\varphi}$ 
with those of the spectral measures $\mu^{(A)}_{\psi,\varphi}$ and $\mu^{(B)}_{\psi,\varphi}$.	
	For that, let $\psi,\varphi\in\mathsf{H}$ and let $\{\psi_k\}_{k\in\mathbb{N}}\subset\mathsf{H}$ be a complete orthonormal basis of $\mathsf{H}$.
	From Trotter product formula it follows that
	\begin{align}
		\nonumber
		\widehat{\mu}^{(\overline{A+B})}_{\psi,\varphi}(t)&=
		\langle\psi| e^{it(\overline{A+B})}\varphi\rangle\\
		\nonumber
		&=\lim_{N\to+\infty}
		\langle\psi|(e^{\frac{it}{N}A}e^{\frac{it}{N}B})^N\varphi\rangle\\
		\nonumber
		&=\lim_{N\to+\infty}\sum_{k_1,\ldots,k_N\in\mathbb{N}}
		\langle\psi|e^{\frac{it}{N}A}\psi_{k_1}\rangle
		\langle\psi_{k_1}|e^{\frac{it}{N}B}\psi_{k_2}\rangle\cdots
		\langle\psi_{k_{N-1}}|e^{\frac{it}{N}A}\psi_{k_N}\rangle
		\langle\psi_{k_N}|e^{\frac{it}{N}B}\varphi\rangle\\
		\label{Eqn: relation between FT of A+B and FT of A,B}
		&=\lim_{N\to+\infty}\sum_{k_1,\ldots,k_N\in\mathbb{N}}
		\widehat{\mu}^{(A)}_{\psi,\psi_{k_1}}\bigg(\frac{it}{N}\bigg)
		\widehat{\mu}^{(B)}_{\psi_{k_1},\psi_{k_2}}\bigg(\frac{it}{N}\bigg)\cdots
		\widehat{\mu}^{(A)}_{\psi_{k_{N-1}},\psi_{k_N}}\bigg(\frac{it}{N}\bigg)
		\widehat{\mu}^{(B)}_{\psi_{k_N},\varphi}\bigg(\frac{it}{N}\bigg)\,.
	\end{align}
	Together with the injectivity of the Fourier transform for finite measure -- \textit{cf.} equation
 \eqref{Eqn: Fourier transform for measure} -- equation \eqref{Eqn: relation between FT of A+B and FT of A,B} provides 
another (more indirect) relation between $P^{(\overline{A+B})}$ and $P^{(A)},P^{(B)}$.
	Notice that Equation \eqref{Eqn: relation between FT of A+B and FT of A,B} can be proved without invoking either 
Lemma \ref{proptec} or Theorem \ref{teorapgen} -- although it can be seen as an application of Theorem \ref{teorapgen}
 for the case $f=1_{\mathbb{R}}=f_\delta\in\mathscr{F}(\mathbb{R})$,
 $\delta$ denoting the Dirac measure 
centred at $0$.} \hfill 
$\blacksquare$ 
\end{remark}

\begin{example}\label{Ex: exactness result for position and momentum operator}
	{\em We provide here  an explicit computation of formula \eqref{mainid2} which is shown to match with the exact
 result -- \textit{cf.} Theorem \ref{teorapgen}
 and part (3) Remark \ref{Rmk: FT of A+B-spectral measure in terms of FT of A,B-spectral measure}.
	Let  $\mathsf{H}=L^2(\mathbb{R}, dx)$ and let $A:X$ be (with obvious domain)  the
 selfadjoint  {\em position operator}  $(X\psi)(x):=x\psi(x)$, and  $B:=P$ the
	selfadjoint {\em momentum operator} (the unique selfadjoint extension of  $(B\psi)(x):=i\frac{d\psi}{dx}(x)$ for $\psi \in \cS(\bR)$)\footnote{We 
recall that our convention for the Fourier
 transform is $\widehat{g}(k):=(2\pi)^{-1/2}\int_{\mathbb{R}}g(x)e^{ixk}dx$, $g(x)=(2\pi)^{-1/2}\int_{\mathbb{R}}\widehat{g}(k)e^{-ikx}dk$ for $g\in\mathscr{S}(\mathbb{R})$ -- \textit{cf.} equation \eqref{f-muf}.}.
	It turns out that $\cS(\bR) \subset D(A)\cap D(B)$ and $aA+bB$ is essentially selfadjoint thereon for every $a,b\in \bR$,  its closure
	being a  generator of the Weyl algebra \cite{Moretti2017} according to Stone's theorem. As a consequence, the
 previously developed theory  can be applied to the pair $A$,$B$.
	
	Taking advantage of {\em Weyl algebra commutation relations}, namely
	$$ e^{iaX}e^{ibP}=e^{-\frac{i}{2}ab}e^{iaX+ibP}, \qquad a, b \in \bR,$$
	we have that for all $\phi,\psi\in\mathsf{H}$
	\begin{align}\label{Eqn: exact result for X,P}
	\widehat{\mu}_{\phi,\psi}^{(\overline{A+B})}(t)=
	\langle\phi|e^{it(\overline{A+B})}\psi\rangle=
	e^{-\frac{i}{2}t^2}\langle\phi|e^{itX}e^{itP}\psi\rangle=
	e^{-\frac{i}{2}t^2}\int_{\mathbb{R}}\overline{\phi(x_1)}e^{itx_1}\psi(x_1-t)dx_1\,.
	\end{align}
	We argue that equation \eqref{mainid2} reproduces the exact result \eqref{Eqn: exact result for X,P}.
	A simple computation  shows that, if $E\in \cB(\bR)$,
	\begin{align}
	(P^{(A)}_E\psi)(x)=1_E(x)\psi(x)\,,\qquad
	\widehat{(P^{(B)}_E\psi)}(k)=1_E(k)\widehat{\psi}(k)\,,\label{eq-XP-2}
	\end{align}
	for, respectively, $\psi \in D(A)$ or $\psi\in D(B)$ and where $\hat{\cdot}$ denotes the Fourier-Plancherel transform.
	With this information we wish to compute the limit as $N\to+\infty$ of
	\begin{multline}\label{eq42}
	\left\langle\phi\left|
	\int_{\bR^{2N}} e^{\frac{it}{N}\sum_{\ell=1}^N(x_\ell+k_\ell)}dP^{(A)}_{x_1}dP^{(B)}_{k_1}\cdots dP^{(A)}_{x_N}dP^{(B)}_{k_N}
	\psi\right.\right\rangle \\ =
	\int_{\mathbb{R}^{2N}}e^{\frac{it}{N}\sum_{\ell=1}^N(x_\ell+k_\ell)}\rho_{\phi,\psi}(x_1,\ldots,k_N)dx_1 dk_1\cdots dx_N dk_N\, , 
	\end{multline}
	where $\rho_{\phi,\psi}(x_1,\ldots,k_N)=\frac{1}{(2\pi)^{N-\frac{1}{2}}}\overline{\phi(x_1)}\prod_{\ell=1}^{N-1}e^{ik_{\ell}(x_{\ell+1}-x_\ell)}e^{-ik_Nx_N}\widehat{\psi}(k_N)$. Equality \eqref{eq42} follows from \eqref{DEFint} and \eqref{eq-XP-2} by applying Fourier-Plancherel transform and its inverse $N$ times.
	Notice that we can avoid the regularization $Q_n$ by considering $\rho$ as a distribution -- there is no integration in $t$.
	A direct computation leads to
	\begin{align*}
	\int_{\mathbb{R}^{2N}}e^{\frac{it}{N}\sum_{\ell=1}^N(x_\ell+k_\ell)}&\rho_{\phi,\psi}(x_1,\ldots,k_N)dx_1\cdots dx_Ndk_1\cdots dk_N\\&=
	\int_{\mathbb{R}^{2N}}\frac{e^{\frac{it}{N}\sum_{\ell=1}^Nx_\ell}
	\overline{\phi(x_1)}}{(2\pi)^{N-1/2}}\prod_{\ell=1}^{N-1}e^{ik_{\ell}(x_{\ell+1}-x_\ell+\frac{t}{N})}e^{-ik_N(x_N-\frac{t}{N})}\widehat{\psi}(k_N)
	dx_1\cdots dx_Ndk_1\cdots dk_N\\&=
	\int_{\mathbb{R}^{N}}e^{\frac{it}{N}\sum_{\ell=1}^Nx_\ell}
	\overline{\phi(x_1)}\prod_{\ell=1}^{N-1}\delta(x_{\ell+1}-x_\ell+\frac{t}{N})\psi(x_N-\frac{t}{N})dx_1\cdots dx_N\\&=
	\int_{\mathbb{R}}e^{itx_1-\frac{i}{2}t^2\frac{N(N-1)}{N^2}}
	\overline{\phi(x_1)}\psi(x_1-t)dx_1\\&
	\xrightarrow{N\to+\infty}
	e^{-\frac{i}{2}t^2}\int_{\mathbb{R}}e^{itx_1}
	\overline{\phi(x_1)}\psi(x_1-t)dx_1\,,
	\end{align*}
	where in the Dirac delta distributions we are  ensured that $x_\ell=x_1-\frac{\ell}{N}t$ for $\ell=2,\ldots, N$.\hfill 
$\blacksquare$}
\end{example}

\subsection{Relation with Feynman integration}\label{Sec: relation with Feynman integration}
The next application proves the close relation between formula (\ref{mainid2}) and the notion of {\em Feynman integral} in {\em phase space}. 
Let us consider now a slight modification of example \ref{Ex: exactness result for position and momentum operator}. 
Let $\mathsf{H}:=L^2(\mathbb{R}, dx)$ and let $(A\psi)(x):=x^2\psi(x)$ and $B:=P^2$ be respectively  the square of the selfadjoint 
position and momentum operators defined in example \ref{Ex: exactness result for position and momentum operator}.  
Further let $a=\Omega ^2 /2$ and  $b=(2m)^{-1}$, with $m,\Omega >0$ constants. It is well known that $A+B$ is essentially selfadjoint over 
$\cS(\bR) \subset D(A)\cap D(B)$ and
 $\overline{A+B}$ coincides to  the harmonic oscillator Hamiltonian operator $H$   \begin{equation}\label{HOH}H=\overline{\frac{P^2}{2m}+\frac{\Omega ^2}{2}X^2}.\end{equation}  
Let us recall that if $T$ is a selfadjoint operator on $\sH$ with spectral measure $P^{(T)}$ and $h:\bR\to \bR$ is a Borel measurable map, the spectral measure 
$P^{(h(T))}$ of the selfadjoint operator $h(T)$ (defined via functional calculus out of  $T$) is related to $P^{(T)}$ by the formula
\begin{equation}\label{change of variales formula}
\int_\bR f(\lambda )dP^{(h(T))}_\lambda =\int_\bR f(h(\lambda ))dP^{(T)}_\lambda
\end{equation}
for any Borel-measurable function $f: \bR \to \bC$.\\
As a consequence, if  $f\in \cF(\bR)$ and using $h : \bR \ni r \mapsto r^2$,   \eqref{mainid2} implies 
$$
f(H)=s\mbox{-}\lim_{N\to \infty }\int_{\bR^{2N}}f\left(\frac{1}{N}\sum_{n=1}^N\left(\frac{k_n^2}{2m}+\frac{\Omega^2}{2}x_n^2\right)\right)dP^{(X)}_{x_1}
dP^{(P)}_{k_1}\cdots dP^{(X)}_{x_N}dP^{(P)}_{k_N}.
$$
If we now choose  the function $f\in \cF(\bR)$ of the form $f(\lambda )=e^{-it\lambda}$, with $t\in \bR$, then for any $\phi,\psi \in \sH$ we have the
 following representation for the matrix elements of  the unitary
 evolution operator $U(t)=e^{-itH}$ 
$$\langle \phi | e^{-itH}\psi \rangle=\lim_{N\to \infty }\int_{\bR^{2N}}e^{-i\sum_{n=1}^N\left(\frac{k_n^2}{2m}+\frac{\Omega^2}{2}x_n^2\right)\frac{t}{N}
}\langle \phi | dP^{(X)}_{x_1}dP^{(P)}_{k_1}\cdots dP^{(X)}_{x_N}dP^{(P)}_{k_N}\psi\rangle.$$
Analogously to the derivation of equality \eqref{eq42},
iterated applications of Fourier transform and its inverse yield
\begin{multline}\label{Feynman1}\langle \phi | e^{-itH}\psi \rangle=\lim_{N\to \infty }\int_{\bR^{2N+1}}
\frac{e^{-i\sum_{n=1}^N\left(\frac{k_n^2}{2m}+\frac{\Omega^2}{2}x_n^2\right)\frac{t}{N}}}{(2\pi)^N} \overline{\phi (x_1)}\psi (x_{N+1})e^{i\sum_{\ell=1}^Nk_\ell(x_{\ell+1}-x_{\ell})}\\
dx_1\cdots dx_{N+1} dk_1 \cdots dk_N, 
\end{multline}
 where we recall that the integrals over $\bR^{2N}$ are just a formal symbol standing for the limit of the regularizing procedure described in
 \eqref{DEFint}, since in this case the distribution \eqref{dist1} is not  associated to a complex measure on $\cB(\bR^{2N})$.
Formula \eqref{Feynman1} admits an interpretation in terms of a {\em phase space Feynman path integral } representation for the time 
evolution operator $U(t)$, as discussed, e.g., in \cite{AlGuMa02,KuGo11}. Indeed, it is possible to look at the exponent appearing in the integral 
on the right-hand side of  \eqref{Feynman1}, namely the function 
 $$S(x_0,\dots, x_N, k_1,\dots , k_N)=\sum_{l=1}^Nk_l(x_l-x_{l-1})-\sum_{n=1}^N\left(\frac{k_n^2}{2m}+\frac{\Omega^2}{2}x_n^2\right)\frac{t}{N},$$
 as the Riemann sum approximation of the classical {\em action functional} in the {\em Hamiltonian formulation}:
 $$S[q,p]=\int_0^t p(s)dq(s)-\int_0^th(q(s),p(s))ds,$$
 where $h:\bR^2\to \bR$ is the classical Hamiltonian of the harmonic oscillator, i.e. 
$$h(x,k)=\frac{k^2}{2m}+\frac{\Omega^2}{2}x^2\:, \quad (x,k)\in \bR^2\:.$$

A completely similar discussion can be repeated in the case the harmonic oscillator potential is replaced with a more general potential, 
namely a measurable map  $V:\bR\to \bR$ such that  the sum of the associated multiplication operator $M_V:D(M_V)\subset \sH \to \sH$ 
and $B=P^2$ is essentially selfadjoint.
 We stress that the limit as $N\to \infty$ on the right-hand side of \eqref{Feynman1} cannot be interpreted in terms of a well-defined
 integral over an infinite dimensional space of paths $(q,p):[0,t]\to \bR^2 $ , formally written (in Feynman path integral notation) as 
 $$\langle \phi | e^{-itH}\psi \rangle= \int_{\{(q,p):[0,t]\to \bR^2\}}  e^{iS[q,p]}\bar \phi 
(q(t))\psi (q(0)) dqdp$$ since it is impossible to construct a corresponding   complex measure, as 
extensively discussed in \cite{ AlMa16} and  in section \ref{Sec: The sequence of measures nu^N_phi,psi  generally does not admit (projective) limit}.
\subsection{(Semi)bounded  $A,B$, Jordan product, Polynomials, Feynman-Kac formula} If the selfadjoint 
operators $A$ and $B$ are bounded then $aA+bB \in \gB(\sH)$ is defined for all $a,b \in \mathbb{R}$ and (\ref{supp-reg-measures}) 
 permits us to extend (\ref{mainid2}) to  the case of $f\in C^\infty(\bR)$. A similar result is valid when $A$ and $B$ are bounded from below.

Let us start with $A,B \in \gB(\sH)$.  For any choice of $a,b\in \bR$ we then have 
\begin{align}
\int _{\bR^{2N}} f\bigg(\frac{1}{N}\sum_{l=1}^{N}(a\lambda _l &+b\mu_l)\bigg)d\nu^{(N,n)}_{\phi,\psi,Q}(\lambda_1, \mu_1,
 \dots , \lambda _N,\mu_N) \nonumber \\
&= \int _{\bR^{2N}} \tilde{f}\left(\frac{1}{N}\sum_{l=1}^{N}(a\lambda _l+b\mu_l)\right)d\nu^{(N,n)}_{\phi,\psi,Q}(\lambda_1, \mu_1,
 \dots , \lambda _N,\mu_N) \label{limnuNnZ0} 
\end{align}
for every $n\in \bN$, if  $\tilde{f}\in \cS(\bR)$ is such that $\tilde{f}(x)=f(x)$ if $x \in [-\|aA\|-\|bB\|, \|aA\|+\|b\|]$. As a consequence,
(\ref{limnuNn}) is still valid with $\tilde{f}$ in place of $f$ on the right-hand side, and the limit does not depend neither on the choice of the
 regularizing sequence $Q:=\{Q_n\}_{n\in \bN}$ nor on the choice of  $\tilde{f}$:
\begin{multline}\label{limnuNnZ}
	\lim_{n\to \infty } \int _{\bR^{2N}} f\bigg(\frac{1}{N}\sum_{l=1}^{N}(a\lambda _l+b\mu_l)\bigg)d\nu^{(N,n)}_{\phi,\psi,Q}(\lambda_1,\mu_1,\dots ,\lambda _N,\mu_N)\\
	=\int_\bR\left\langle \phi \left|  (e^{i \frac{t}{N}aA}e^{i \frac{t}{N}bB})^N\psi\right.\right\rangle d\nu _{\tilde{f}}(t)\,.
	\end{multline}
A similar result is found  when $A$ and $B$ are (unbounded) selfadjoint operators which are bounded from below and $a,b \geq 0$
are such that $aA+bB$ is essentially selfadjoint.
In this case  $\overline{aA+bB}$ is bounded from below as well with
	$\inf \sigma(\overline{aA+bB}) \geq a\inf \sigma(A) + b \inf \sigma(B)\,.$
With the said $A,B,a,b$, both (\ref{limnuNnZ0}) and (\ref{limnuNnZ}) are true provided
$f \in C^\infty(\bR)$, $\tilde{f} \in \cS(\bR)$, and $f(x) =\tilde{f}(x)$ for $x \geq a\inf \sigma(A) + b \inf \sigma(B)$.

\noindent Identity (\ref{limnuNnZ}) is  sufficient in both cases for giving the following definition  based on  Riesz' lemma.
\begin{definition}\label{defext} {\em Assume that one of the following two cases holds
\begin{itemize}
\item[{\bf (a)}] $A,B \in \gB(\sH)$ are selfadjoint,  $a,b \in \bR$, and  $f \in C^\infty(\bR)$;
\item[{\bf (b)}]  $A,B$ are selfadjoint and bounded from below,  $a,b\geq 0$,  $aA+bB$ is essentially selfadjoint, and $f\in C^\infty(\mathbb{R})$ is 
such that there exists $\tilde{f}\in\mathscr{S}(\mathbb{R})$ for which $f(x)=\tilde{f}(x)$ for $x\geq a\inf\sigma(A)+b\inf\sigma(B)$.
\end{itemize}
 Then 
$$\int_{\bR^{2N}} f\left( \frac{1}{N}\sum^N_{n=1}  (a\lambda_{n} +b \mu_{n})\right) dP^{(A)}_{\lambda_1}dP^{(B)}_{\mu_1}\cdots dP^{(A)}_{\lambda_N}dP^{(B)}_{\mu_N} \in \gB(\sH)$$
is defined as the unique operator satisfying 
(\ref{DEFint}).} \hfill 
$\blacksquare$ 
\end{definition}
 \noindent This definition is an evident extension of  the definition for $f\in \cF(\mathbb{R})$ appearing in Theorem \ref{defint} 
  because,  if  $\tilde{f}\in \cS(\mathbb{R})\subset \cF(\mathbb{R})$ is defined as above with respect to $f$, then
\begin{align} \int_{\bR^{2N}} f\bigg( \frac{1}{N}\sum^N_{n=1}  (a\lambda_{n} &+b \mu_{n})\bigg) dP^{(A)}_{\lambda_1}dP^{(B)}_{\mu_1}\cdots dP^{(A)}_{\lambda_N}dP^{(B)}_{\mu_N} \nonumber \\=
&\int_{\bR^{2N}} \tilde{f}\left( \frac{1}{N}\sum^N_{n=1}  (a\lambda_{n} +b \mu_{n})\right) dP^{(A)}_{\lambda_1}dP^{(B)}_{\mu_1}\cdots dP^{(A)}_{\lambda_N}dP^{(B)}_{\mu_N} \end{align}
The remaining properties (i)-(iv) of Theorem \ref{defint} hold as well
provided $\|f\|_\cF$ is replaced for $\|\tilde{f}\|_\cF$ on the right-hand side of (\ref{ineq2}) in (ii).\\
With this extended definition,  Theorem \ref{teorapgen} is still valid as the reader immediately proves noticing in
 particular that from (\ref{Eq: function of self-adjoint operator}) and (\ref{spectral2}), we have
$$f(aA+bB) = \tilde{f}(aA+bB)\:.$$
\begin{proposition}\label{proplast}
 If either (a) or (b) in Definition \ref{defext}  is valid,  then 
\begin{equation}\label{mainid25}
f\left(\overline{aA+bB}\right)  = \mbox{s-}\lim_{N\to +\infty} \int_{\bR^{2N}} f\left( \frac{1}{N}\sum^N_{n=1}  (a\lambda_{n} + b\mu_{n})\right) dP^{(A)}_{\lambda_1}dP^{(B)}_{\mu_1}\cdots dP^{(A)}_{\lambda_N}dP^{(B)}_{\mu_N}\:.
\end{equation}
\end{proposition}

\noindent  If with $A,B\in \gB(\sH)$, definition \ref{defext} applies in particular to the operator $aA+bB$ as in  (\ref{Jordan1}) below,
  when referring to  the smooth map $f: \bR \ni x \mapsto x \in \bR$.
That  identity shows how the real linear space structure of the algebra of bounded observables can be  constructed 
employing the PVMs of the involved observables only. The whole algebra of bounded observables $\gB(\sH)$ has another natural
 operation that is the (non-associative)  {\em Jordan product}:
$$A \circ B := \frac{1}{2}(AB+BA)\:, \quad\mbox{for $A=A^*$, $B=B^*$, $A,B\in \gB(\sH)$. }$$
The Jordan product of $A$ and $B$ admits an expansion similar to (\ref{Jordan1}) where, again, only the PVMs of $A$ and $B$ is used. 
These  results can be further generalized as stated in the last item of the following theorem which concerns polynomials.
\begin{theorem} \label{teopolynomials}
	Let $A, B\in \gB(\sH)$ be selfadjoint operators, with $\sH$ separable, whose PVMs are respectively denoted
 by $P^{(A)}$ and $P^{(B)}$. 
 Let  $\nu^{(N,n)}_{\phi,\psi,Q} :\cB(\bR^{2N})\to \bC$ be the complex measures   defined as in (\ref{reg-measures}).
 The following facts are true.
\begin{itemize}
\item[{\bf(a)}] If $a,b \in \bR$,
\begin{equation}\label{Jordan1}
aA+bB  = \mbox{s-}\lim_{N\to +\infty} \int_{\bR^{2N}} \left(a\sum^N_{n=1} \frac{\lambda_n}{N} + b\sum^N_{n=1} \frac{\mu_n}{N} \right)\: dP^{(A)}_{\lambda_1}dP^{(B)}_{\mu_1}\cdots dP^{(A)}_{\lambda_N}dP^{(B)}_{\mu_N}\:.
\end{equation}
\item[{\bf(b)}]  The Jordan product of $A$ and $B$ satisfies
\begin{equation}\label{Jordan2}
A\circ B = \mbox{s-}\lim_{N\to +\infty} \int_{\bR^{2N}} \left(\sum^N_{n=1} \frac{\lambda_{n}}{N}\right)
\left( \sum^N_{n=1}  \frac{\mu_{n}}{N}\right) \:\: dP^{(A)}_{\lambda_1}dP^{(B)}_{\mu_1}\cdots dP^{(A)}_{\lambda_N}dP^{(B)}_{\mu_N}\:,
\end{equation}
where the operator after the symbol of limit  is the unique operator such that
\begin{multline}
\left\langle \phi \left| \int_{\bR^{2N}}\left(\sum^N_{n=l} \frac{\lambda_{l}}{N}\right)  
\left(\sum^N_{n=l} \frac{\mu_{l}}{N}\right)  \:\:dP^{(A)}_{\lambda_1}dP^{(B)}_{\mu_1}\cdots dP^{(A)}_{\lambda_N}dP^{(B)}_{\mu_N}\psi \right. \right\rangle \\
=\lim_{n\to \infty } \int _{\bR^{2N}} \left(\sum_{l=1}^{N}\frac{\lambda _l}{N}\right)\left(\sum_{l=1}^{N}\frac{\mu _l}{N}\right)\:\: d\nu^{(N,n)}_{\phi,\psi,Q}(\lambda_1, \mu_1, \dots , \lambda _N,\mu_N)\:,\label{DEFintJ}
\end{multline}
for all $\phi, \psi\in \sH$.
\item[{\bf(c)}]  Assume  $p,q=0,1,2,\ldots$ and let $S_{p+q}(A^p B^q)$ denote the selfadjoint operator consisting of the symmetrized product of $p$ copies of $A$ and $q$ copies of $B$,
$$S_{p+q}(A^p B^q) := {p+q\choose q}^{-1}(A^pB^q+ A^{p-1}BAB^{q-1}+ \cdots + B^qA^p)\:.$$
If  only a finite number of  coefficients $c_{p,q}\in \bR$ do not vanish, then
\begin{align}\label{Jordan2n}
&\sum_{p,q} c_{p,q}S_{p+q}(A^p B^q) \nonumber \\ &= \mbox{s-}\lim_{N\to +\infty} \int_{\bR^{2N}} \sum_{p,q} c_{p,q}\left(\sum^N_{n=1} \frac{\lambda_{n}}{N}\right)^p
\left( \sum^N_{n=1}  \frac{\mu_{n}}{N}\right)^q \:\: dP^{(A)}_{\lambda_1}dP^{(B)}_{\mu_1}\cdots dP^{(A)}_{\lambda_N}dP^{(B)}_{\mu_N}\:,
\end{align}
where the operator after the symbol of limit  is the unique operator such that
\begin{multline}
\left\langle \phi \left| \int_{\bR^{2N}}\sum_{p,q} c_{p,q}\left(\sum^N_{l=1} \frac{\lambda_{l}}{N}\right)^p
\left(\sum^N_{l=1} \frac{\mu_{l}}{N}\right)^q  \:\:dP^{(A)}_{\lambda_1}dP^{(B)}_{\mu_1}\cdots dP^{(A)}_{\lambda_N}dP^{(B)}_{\mu_N}\psi \right. \right\rangle \\=\lim_{n\to \infty } \int _{\bR^{2N}} \sum_{p,q} c_{p,q}\left(\sum_{l=1}^{N}\frac{\lambda _l}{N}\right)^p\left(\sum_{l=1}^{N}\frac{\mu _l}{N}\right)^q\:\: d\nu^{(N,n)}_{\phi,\psi,Q}(\lambda_1, \mu_1, \dots , \lambda _N,\mu_N)\:,\label{DEFintJn}
\end{multline}
for all $\phi, \psi\in \sH$.
\end{itemize}
The limit in the right-hand side of (\ref{DEFintJ}) and (\ref{DEFintJn}) does not depend on the choice of the regularizing sequence $Q$.
\end{theorem}

\begin{proof} We only have to prove (b) and (c). Regarding (b), 
it essentially relies on the elementary decomposition $xy = \frac{1}{2}(x+y)^2-\frac{1}{2}x^2-\frac{1}{2}y^2$, which implies
\begin{align}&2\int _{\bR^{2N}}  \left(\sum_{l=1}^{N}\frac{\lambda _l}{N}\right)\left(\sum_{l=1}^{N}\frac{\mu _l}{N}\right)\:\: d\nu^{(N,n)}_{\phi,\psi,Q}
(\lambda_1, \mu_1, \dots , \lambda _N,\mu_N)\nonumber \\
&= \int _{\bR^{2N}}\left( \frac{1}{N}\sum_{l=1}^{N}(\lambda _l +\mu_l)\right)^2\:\: d\nu^{(N,n)}_{\phi,\psi,Q}(\lambda_1, \mu_1, \dots , \lambda _N,\mu_N)\nonumber \\
& -\int _{\bR^{2N}} \left(\frac{1}{N}\sum_{l=1}^{N} \lambda_l\right)^2 \:\: d\nu^{(N,n)}_{\phi,\psi,Q}(\lambda_1, \mu_1, \dots , \lambda _N,\mu_N)
 -\int _{\bR^{2N}} \left(\frac{1}{N}\sum_{l=1}^{N} \mu_l\right)^2 \:\: d\nu^{(N,n)}_{\phi,\psi,Q}(\lambda_1, \mu_1, \dots , \lambda _N,\mu_N)\:.
\end{align}
As discussed above, the limits for $n\to +\infty$ of the final three addends exist and do not depend on the used sequence $Q$.
We conclude that the right-hand side of (\ref{DEFintJ}) exists and does not depend on $Q$. More precisely, from (\ref{limnuNnZ}),
\begin{align}\label{NIFE}& \lim_{n\to +\infty} \int _{\bR^{2N}}  \left(\sum_{l=1}^{N}\frac{\lambda _l}{N}\right)\left(\sum_{l=1}^{N}\frac{\mu _l}{N}\right) \:\: d\nu^{(N,n)}_{\phi,\psi,Q}(\lambda_1, \mu_1, \dots , \lambda _N,\mu_N) \nonumber\\
 &=\frac{1}{2}\int_\bR\left\langle \phi \left|  (e^{i \frac{t}{N}A}e^{i \frac{t}{N}B})^N\psi\right.\right\rangle d\nu _{\tilde{f}}(t)
- \frac{1}{2}\int_\bR\left\langle \phi \left|  e^{i tA}\psi\right.\right\rangle d\nu _{\tilde{f}}(t)- \frac{1}{2}\int_\bR\left\langle \phi \left|  e^{i tB}\psi\right.\right\rangle d\nu _{\tilde{f}}(t)\end{align}
where $\tilde{f}:\bR \to \bR$ is such that $\tilde{f}\in \cS(\bR)$ such that 
$\tilde{f}(x) = x^2$ for the values $x \in [-\|A\|-\|B\|, \|A\|+\|B\|]$.  At this point Riesz' lemma proves that there is a unique operator 
satisfying (\ref{DEFintJ}).  Decomposing
$$ \left(\sum_{l=1}^{N}\frac{\lambda _l}{N}\right)\left(\sum_{l=1}^{N}\frac{\mu _l}{N}\right) = \frac{1}{2}\left(\frac{1}{N}\sum_{l=1}^{N}(\lambda_\ell+
 \mu_\ell)\right )^2 - \frac{1}{2}\left(\frac{1}{N}\sum_{l=1}^{N}\lambda_\ell
 \right )^2 -\frac{1}{2}
 \left(\frac{1}{N}\sum_{l=1}^{N}
 \mu_\ell\right )^2$$
on the right-hand side of (\ref{Jordan2}),
and noticing that analogously to $xy = \frac{1}{2}(x+y)^2-\frac{1}{2}x^2-\frac{1}{2}y^2$,
$$A \circ B = \frac{1}{2}(A+B)^2 -  \frac{1}{2}A^2- \frac{1}{2} B^2\:,$$
 (\ref{Jordan2}) is an immediate consequence of (\ref{mainid25}) applied to the three addends on the right-hand sides of the two identies above. The proof of (c) is a straightforward generalization the proof   of (b). It relies on the well-known Sylvester's result (see, e.g., the review part of \cite{Rez13}), that every real (or complex)  homogeneous polynomial in $x$ and $y$ of degree $m$, in particular the monomial $x^py^q$ (with $m:=p+q)$, can be decomposed as $x^py^q = \sum_{k} (a_kx+b_ky)^m$ for a finite number ($\leq m+1$) of constants  $a_k,b_k \in \bR$ (resp. $\bC$), and the fact that this result easily implies
$S_{p+q}(A^pB^q) = \sum_k (a_kA+b_kB)^m$.
\end{proof}

\begin{remark} {\em Item (c) proves the existence of a linear map associating each real  polynomial $p$ on $\bR^2$ to a {\em suitably symmetrized} (as in the left-hand side of (\ref{Jordan2n})) polynomial of the operators $A$ and $B$. One naturally wonders if that 
map is a {\em quantization map} \cite{landsman}. It is not possible to answer the question without further information on the $^*$-algebra generated by $A$ and $B$.
 As a matter of fact, a necessary condition is that the aforementioned map transforms a Poisson structure on the space of polynomials to a real Lie-algebra structure on the space of the said polynomials of $A$ and $B$,  whose Lie bracket is the commutator of operators (more precisely $i[\cdot, \cdot]$). In the general case, the space of those polynomials of operators is not a Lie algebra, unless $A$ and $B$ satisfy some further relations,  in particular, it must be possible to expand $i[A,B]$ as a real polynomial of $A$ and $B$ which is symmetrized as said.} \hfill $\blacksquare$
\end{remark}

\noindent Case (b) in Proposition  \ref{proplast}  has  an interesting application which proves the interplay of formula \eqref{mainid2} and  the so-called 
 {\em Feynman-Kac formula}.
\begin{example}{\em
Let us consider again the setting of section \ref{Sec: relation with Feynman integration}. 
Since in this case both operators $A$ and $B$ have positive spectrum, we can apply formula \eqref{mainid25} also for $f$ with 
exponential growth. In particular we can choose $f:\bR\to \bR$ of the form $f(\lambda)=e^{-\tau\lambda}$, with $\tau > 0$, obtaining the following representation for the matrix elements of the semigroup generated by the selfadjoint operator $H$ (see Eq. \eqref{HOH}):

\begin{multline}\label{heat1}
	\langle \phi | e^{-\tau H}\psi \rangle=
	\lim_{N\to \infty }\frac{1}{(2\pi)^N}\int_{\bR^{2N+1}}
	e^{-\sum_{n=1}^N\left(\frac{k_n^2}{2m}+\frac{\Omega^2}{2}x_n^2\right)\frac{\tau}{N}}
	\overline{\phi (x_1)}\psi (x_{N+1})e^{i\sum_{\ell=1}^Nk_\ell(x_{\ell+1}-x_\ell)} \\
	dx_1\cdots dx_{N+1} dk_1 \cdots dk_N\,.
\end{multline}
If $\psi \in L^1(\bR,dx)$, the integrals on $\bR^{2N+1}$ appearing on the right-hand side are  absolutely convergent thanks
 to the fast decaying properties of the exponentials appearing in the integrand. In particular, in this case, by applying Fubini 
Theorem and integrating with respect to $k_1,\ldots, k_N$, we get
\begin{equation}\label{heat2}\langle \phi | e^{-\tau H}\psi \rangle=\lim_{N\to \infty }
\left( \frac{mN}{2\pi \tau }\right)^{N/2}\int_{\bR^{N+1}}
e^{-\sum_{n=1}^N\frac{\Omega^2}{2}x_n^2\frac{\tau }{N}}\overline{\phi (x_1)}\psi (x_{N+1})e^{-\frac{m}{2}\sum_{\ell=1}^N\frac{(x_{\ell+1}-x_{\ell})^2}{\tau/N}} 
dx_1\cdots dx_{N+1}\,.
\end{equation}
In this case, thanks to the presence of the Gaussian density $$\rho _G(x_1,...,x_N)=\left( \frac{mN}{2\pi \tau }\right)^{N/2} e^{-\frac{m}{2}\sum_{l=1}^N\frac{(x_l-x_{l-1})^2}{\tau/N}} $$ in the integrand,  the limit for $N\to\infty$ of the finite dimensional integrals appearing in the right-hand side of \eqref{heat2} can be interpreted in terms of an integral over the Banach space $C_\tau(\bR)$ of continuous paths $\omega :[0,\tau]\to \bR$ endowed with the uniform norm, the corresponding Borel $\sigma$-algebra and the Wiener Gaussian measure $W$.  By a standard argument (see, e.g.,  \cite{Sim}) formula \eqref{heat2} yields the celebrated {\em Feynman-Kac formula}:
$$\langle \phi | e^{-\tau H}\psi \rangle=\int_\bR\overline{\phi (x)}\int_{C_\tau (\bR)}\psi(x+\omega (\tau )) e^{-\frac{\Omega^2}{2}\int_0^\tau (\omega (s)+x)^2ds} dW(\omega)dx.$$

As remarked at the end of section \ref{Sec: relation with Feynman integration}, the discussion above generalizes to the case the 
harmonic oscillator potential is replaced with a Borel measurable non-negative map $V:\bR\to \bR$ such that the sum of
 the associated multiplication operator $M_V:D(M_V)\subset\sH \to \sH$ and $B=P^2$ is essentially selfadjoint. Related ideas can be also found in \cite{JL,JoLaNi,Nie} as well as in Lapidus'  papers from the late 80's \cite{Lap1,Lap2,Lap3}.
}\hfill $\blacksquare$\end{example}

\section{Counterexamples}\label{Sec: counterexamples}

In this section, we provide a series of counterexamples which make evident  the rigidity of the structures involved in the rigorous generalization of formula
 \eqref{Eq: f(aA+bB) in the finite dimensional case} to infinite dimensional case  -- \textit{cf.} Theorem \ref{teorapgen}.

In section \ref{Sec: extension to infinite dimensional Hilbert spaces}, we introduced a regularization procedure through the projectors $Q_n$  to safely
 define a family of measures $\nu^{(N,n)}_{\phi,\psi}$ whose distributional limit $n\to+\infty$ permitted us  to define the operator 
integral in equation \eqref{Eq: operator integral w.r.t. product of 2N-projectors} -- \textit{cf.} Theorem \ref{defint}.
This may appear artificial at first glance and one may wonder if  the regularization $Q_n$ can be removed. This is not the
 case -- \textit{cf.} section \ref{Sec: The Q_n-regularization  is generally necessary}.

Moreover, even in the finite dimensional case -- where $Q_n$ can be chosen to be the identity and $\nu^{(N)}_{\phi,\psi}$ does 
exists as a complex measure on $\mathbb{R}^{2N}$ -- the properties of the sequence of measures $\lbrace\nu^{(N)}_{\phi,\psi}\rbrace_N$ 
are however not sufficient  to define a suitable projective limit $N\to +\infty$ -- \textit{cf.} section \ref{Sec: The sequence of measures nu^N_phi,psi  generally does not admit (projective) limit}.

Finally,  the convergence properties of $\mu^{(N,n)}_{\phi,\psi}$ as $n\to+\infty$ are considered in the best case scenario, namely when
 $P^{(A)}$, $P^{(B)}$ commute.
Even in this, in a sense tamed situation,  simple cases pop out where $\mu^{(N,n)}_{\phi,\psi}$
 does not converge in the space of complex measure $\mathscr{M}(\mathbb{R}^{2N})$ -- \textit{cf.}
 section \ref{Sec: Convergence of N,n-sequence of measure in the space of complex measure for commuting PVMs: a counterexample}.

\subsection{The $Q_n$-regularization  is generally necessary}\label{Sec: The Q_n-regularization  is generally necessary}
 Lemma \ref{proptec} introduces the regularized measure $\nu^{(N,n)}_{\phi,\psi}$ which is built out of the PVMs $P^{(A)}$, $P^{(B)}$ 
associated with the operators $A$, $B$ together with a suitable ``regularizing" sequence of projectors $\{Q_n\}$.
	One may wonder whether it is possible to avoid the regularization $\{Q_n\}$, thus considering a complex measure 
$\nu^{(N)}_{\phi,\psi}$ defined on products of Borel set $I_\ell,J_\ell\in\mathscr{B}(\mathbb{R})$ as
	\begin{align*}
		\nu^{(N)}_{\phi,\psi}(\times_{\ell=1}^NI_\ell\times J_\ell):=
		\langle\phi | P^{(A)}_{I_1}P^{(B)}_{J_1}\cdots P^{(A)}_{I_N}P^{(B)}_{J_N}\psi\rangle\,.
	\end{align*}
	Unfortunately $\nu^{(N)}_{\phi,\psi}$ generally  does not extend to a complex measure on $\mathbb{R}^{2N}$
 as is in particular  shown by the following physically relevant example.

\begin{example}\label{Ex: counterexample to well-definiteness of complex measure}
{\em 	With reference to example \ref{Ex: exactness result for position and momentum operator},
 let us check if, in this concrete case, $\nu^{(N)}_{\phi,\psi}$ extends to a complex measure on $\mathbb{R}^{2N}$.
	
	For all $I_1,J_1,\ldots,I_N,J_N\in\mathscr{B}(\mathbb{R})$ we have
	\begin{align*}
	\nu^N_{\phi,\psi}(\times_{\ell=1}^N I_\ell\times J_\ell)&=
	\int1_{I_1}(x_1)1_{J_1}(k_1)\cdots 1_{J_N}(k_N)\rho_{\phi,\psi}(x_1,k_1,\ldots,k_N)\mathrm{d}x_1\cdots\mathrm{d}k_N\,,\\
	\rho_{\phi,\psi}(x_1,\ldots,k_N)&:=(2\pi)^{-N+1/2}\overline{\phi(x_1)}\prod_{\ell=1}^{N-1}e^{ik_{\ell}(x_{\ell+1}-x_\ell)}e^{-ik_Nx_N}\widehat{\psi}(k_N)\,.
	\end{align*}
	Since $\rho_{\phi,\psi}\in L^1_{\mathrm{loc}}(\mathbb{R}^{2N},d^Nxd^Nk)\setminus
 L^1(\mathbb{R}^{2N},d^Nxd^Nk)$ is it evident  that $\nu^N_{\phi,\psi}$ does not extend to a
 well-defined complex measure on $\mathbb{R}^{2N}$ because the total variation of such complex measure 
would be $||\rho_{\phi,\psi}||_{L^1(\mathbb{R}^{2N},d^Nxd^Nk)}$ which is not finite. \hfill 
$\blacksquare$ }
\end{example}

\subsection{The sequence of measures $\nu^{(N)}_{\phi,\psi}$  generally does not admit a (projective) limit}
\label{Sec: The sequence of measures nu^N_phi,psi  generally does not admit (projective) limit}

Let us consider the elementary case  of a  finite-dimensional Hilbert space, where the regularization 
driven by $Q$ is not necessary. In fact,  we can simply fix $Q_n = I$. 

 	A natural issue  is whether or not it is possible to interpret the limit for $N\to \infty$ on the right-hand side of  \eqref{Eq: f(aA+bB) in the finite
 dimensional case} in terms of a integral on an infinite dimensional space.
 	Indeed, in  the particular case where $\sH$ is finite dimensional, for any $\phi,\psi\in \sH$ it is possible to define 
a family of complex Borel measures $\{\nu^{(N)}_{\phi,\psi}\}_{N\in \mathbb{N}}$, with $\nu^{(N)}_{\phi,\psi}:\mathscr{B}(\bR^{2N})\to \bC$.
 	By introducing for any $M,N\in \bN$ with $M\leq N$, the  projection maps $\pi^N_M:\bR^{2N}\to \bR^{2M}$, it is simple to see that the measures $\{\nu^{(N)}_{\phi,\psi}\}_{N\in \mathbb{N}}$  fulfil the  consistency condition 
 	$$ \nu^{(M)}_{\phi,\psi}=(\pi^N_M)_*\nu^{(N)}_{\phi,\psi}, \qquad M\leq N,$$
 	that is for all $I_1,...,I_M, J_1,...,J_M\in \mathscr{B}(\bR)$ 
 	$$\nu^{(M)}_{\phi,\psi}(\times_{\ell=1}^M I_\ell\times J_\ell)=\nu^{(N)}_{\phi,\psi}(\times_{\ell=1}^MI_\ell\times J_\ell     \   \times   \   (\bR \times \bR)^{N-M}).$$ 
  	It is then meaningful to investigate whether the projective family of complex measures $\{\nu^{(N)}_{\phi,\psi}\}_{N\in \mathbb{N}}$  
admits a projective limit.
  	For that, let us consider the infinite dimensional space $(\bR\times \bR)^\bN$ of all sequences $\{(\lambda _n,\mu_n)\}_{n \in \mathbb{N}}$.
  	Moreover, let $\Sigma$ be the $\sigma$-algebra on $(\bR\times \bR)^\bN$ generated by the cylinder sets $\pi_N^{-1}(E)$, 
where $N\in\mathbb{N}$, $E\in  \mathscr{B}(\bR^{2N})$ and $\pi_N:(\bR\times \bR)^\bN\to (\bR\times \bR)^N$ is the canonical projection map.
	Within this setting,  it is meaningful to investigate  the existence of a complex bounded measure $\nu_{\phi,\psi}$ 
on $( (\bR\times \bR)^\bN, \Sigma)$ such that $\nu^{(N)}_{\phi,\psi}=(\pi_N)_*\nu_{\phi,\psi} $ for all $n\in \bN$ (see appendix 
\ref{App: projective limit of complex measure} for further details).
 However, as proved in Appendix \ref{App: projective limit of complex measure}, a necessary condition for the existence
 of $\nu_{\phi, \psi}$ is a uniform bound on the total variation $\|\nu^{(N)}_{\phi,\psi}\|=|\nu^{(N)}_{\phi,\psi}(\mathbb{R}^{2N})|$
 of the measures $\{\nu^{(N)}_{\phi,\psi}\}_{N\in \mathbb{N}}$.
Unfortunately in general this does not hold when $A$ and $B$ do not commute as the following example shows.

\begin{example}\label{Ex: counterexample Pauli matrices}
{\em Let us consider the special  case  $\sH=\bC^2$, where   $A:=\sigma _z$ and $B:=\sigma _x$, $\sigma _z, \sigma _x$ denoting two {\em Pauli matrices}. 
In this case,  the Hilbert space is finite dimensional and for any $N\geq 1$ and $\phi, \psi \in \bC^2$, it is possible to define the complex Borel
 measure $\nu^N_{\phi, \psi}$ on $\{\pm 1\}^{2N}$ (the topology being the discrete one) as:
$$\nu^N_{\phi, \psi}(\{(a_1,b_1,...,a_N,b_N)\}):=\langle \phi | P^{(A)}_{a_1}P^{(B)}_{b_1}\cdots P^{(A)}_{a_N}P^{(B)}_{b_N}\psi \rangle$$
where $a_1,...,a_N, b_1,..., b_N\in \{\pm1\}$ and $P^{(A)}_{\pm}, P^{(B)}_{\pm}$ are defined as
$$P^{(A)}_{\pm}:=|z_\pm\rangle \langle z_\pm|, \qquad P^{(B)}_{\pm}:=|x_\pm\rangle \langle x_\pm|,$$
where $|z_+\rangle, |z_-\rangle$, resp.$|x_+\rangle, |x_-\rangle$, are the normalized eigenvectors of $\sigma_z$, resp. $\sigma _x$.

The {\em total variation} of $\nu^{(N)}_{\phi, \psi}$ can be computed explicitly here. 
Indeed, by writing the vectors $\phi, \psi $ as linear combinations of the vectors of the orthonormal basis $ \{|z_+\rangle, |z_-\rangle\}$ and $ \{|x_+\rangle, |x_-\rangle\}$, namely:
$$\phi=\alpha_+|z_+\rangle+\alpha_- |z_-\rangle, \qquad \psi =\beta_+|x_+\rangle+\beta_- |x_-\rangle$$
we get for any $(a_1,b_1,...,a_N,b_N)\in \{\pm 1\}^{2N}$:
$$\left|  \nu^{(N)}_{\phi, \psi}(\{(a_1,b_1,...,a_N,b_N)\})\right|=2^{-N+1/2}(|\alpha _+|+|\alpha _-|)(|\beta _+|+|\beta _-|)\:.$$
Hence we obtain 
$$\left|  \nu^{(N)}_{\phi, \psi}\right|=\sqrt 2\ 2^{N-2}(|\alpha _+|+|\alpha _-|)(|\beta _+|+|\beta _-|)\:.$$
The sequence of mesures $\nu^{(N)}_{\phi,\psi}$ {\em cannot} admit a projective limit, since $\sup_N \left|  \nu^{(N)}_{\phi, \psi}\right|=+\infty$.} $\hfill \blacksquare$ 
\end{example}

\noindent The illustrated  no-go result is strictly related to the non-commutativity of the {\em spectral projectors} 
$P^{(A)}$ and $P^{(B)}$.
 In fact, in the general case of a separable Hilbert space $\sH$, if the spectral measures $P^{(A)}$ and $P^{(B)}$
 commute then,  for any couple of normalized vectors
 $\phi, \psi\in \sH$ and for any $N\geq 1$, the formula  
$$\nu^{(N)}_{\phi, \psi}(\times_{\ell=1}^N I_\ell\times J_\ell)=\langle\phi | P^{(A)}_{I_1}P^{(B)}_{J_1}\cdots P^{(A)}_{I_N}P^{(B)}_{J_N}\psi\rangle$$
defines a complex measure with finite total variation.
In particular, since 
$$   \langle\phi, P^{(A)}_{I_1}P^{(B)}_{J_1}\cdots P^{(A)}_{I_N}P^{(B)}_{J_N}\psi\rangle= \langle\phi| P^{(A)}_{\cap_{l=1}^NI_l}P^{(B)}_{{\cap_{l=1}^NJ_l}}\psi\rangle,$$
the following identity  holds for any set $E\in \cB(\bR^{2N})$,
\begin{equation}\label{inequality-totalvariation}
\nu^{(N)}_{\phi, \psi}(E)=\int_{\bR^2}1_{E}(\lambda_1, \mu_1,\dots \lambda_1, \mu_1)d\nu^1_{\phi, \psi}(\lambda_1, \mu_1),
\end{equation}
where $\nu^{(1)}_{\phi, \psi}(I\times J)=\langle\phi| P^{(A)}_{I}P^{(B)}_{J}\psi\rangle$
 is associated to the joint spectral measure of $A $ and $B$. By \eqref{inequality-totalvariation} it is simple to obtain the following
 inequality for any $N\geq 1$
\begin{equation}\label{un-bound}|\nu^{(N)}_{\phi, \psi}|\leq |\nu^{(1)}_{\phi, \psi}|\leq \|\phi\|  \| \psi\|  .\end{equation}
In addition, in  the particular case where   $\phi=\psi$ and $\|\phi\|=1$ formula \eqref{Eq: f(aA+bB) in the finite dimensional case}
 admits a simple probabilistic interpretation. Indeed, in this case the measures $\{\nu^{(N)}_{\phi, \psi}\}_N$
 form a consistent family of probability measures that can be interpreted as the joint distributions of a sequence 
$\{\xi_n,\eta_n\}_n$ of random variables with values in $\bR^2$ which are  in fact trivially
 completely correlated, i.e. $(\xi_n,\eta_n)=(\xi_1,\eta_1)$ for all $n\in \bN$, and the distribution of the random
 vector $(\xi_1,\eta_1)$ is given by the joint spectral measure of $A$ and $B$. 

In the case the commutativity condition is {\em not} fulfilled, then a uniform bound as \eqref{un-bound} cannot be obtained in general as the folloing example shows.

\begin{example}\label{Ex: counterexample Pauli matrices2}
{\em Let us consider again $\sH:=\bC^2$ as before, and the two spectral measures $P^{(A)}$ and $P^{(B)}$ on $\{\pm 1\}$ defined as
$$P^{(A)}_{\pm}:=|z_\pm\rangle \langle z_\pm|, \qquad P^{(B)}_{\pm}:=|v_\pm\rangle \langle v_\pm|,$$
where $$|v_+\rangle=(\cos \epsilon ) |z_+\rangle+(\sin \epsilon ) |z_-\rangle, \quad |v_-\rangle=-(\sin\epsilon ) |z_+\rangle+(\cos \epsilon )|z_-\rangle\,,\qquad\epsilon>0\,.$$
Let us choose $\phi:=\frac{\sqrt 2}{2} |z_+\rangle+\frac{\sqrt 2}{2} |z_-\rangle$ and $\psi:=\frac{\sqrt 2}{2} |v_+\rangle+\frac{\sqrt 2}{2} |v_-\rangle$
 and consider 
again the complex measure   $\nu^{(N)}_{\phi, \psi}$ on $\{\pm 1\}^{2N}$ defined  for any $(a_1,b_1,...,a_N,b_N)\in \{\pm 1\}^{2N}$ by
$$\nu^{(N)}_{\phi, \psi}(\{(a_1,b_1,...,a_N,b_N)\}):=\langle \phi| P^{(A)}_{a_1}P^{(B)}_{b_1}\cdots P^{(A)}_{a_N}P^{(B)}_{b_N}\psi \rangle. $$
It is simple to verify that 
$$|\nu^{(N)}_{\phi, \psi}|=\sum_{(a_1,b_1,...,a_N,b_N)\in \{\pm 1\}^{2N}}|\nu^{(N)}_{\phi, \psi}(\{(a_1,b_1,...,a_N,b_N)\})|=\frac{1}{2}(|\cos \epsilon|+|\sin \epsilon |)^{2N-1}$$
which shows that $\sup_N||\nu^{(N)}_{\phi, \psi}||=+\infty$, unless the spectral measures commute.} \hfill 
$\blacksquare$ 
\end{example}

\subsection{Convergence of $\{\mu^{(N,n)}_{\phi,\psi}\}_{n \in \mathbb{N}}$ in $\mathscr{M}(\mathbb{R}^{2N})$ for commuting PVMs: a counterexample}
\label{Sec: Convergence of N,n-sequence of measure in the space of complex measure for commuting PVMs: a counterexample}

Let us consider once again the sequence of finite complex measure $\nu^{(N,n)}_{\phi,\psi}$ on $\mathbb{R}^{2N}$.
Since $Q_n\to I$ strongly it follows that, for all $I_\ell,J_\ell\in\mathscr{B}(\mathbb{R}^N)$,
\begin{align*}
	\lim_{n\to+\infty}\nu^{(N,n)}_{\phi,\psi}(\times_{\ell=1}^NI_\ell\times J_\ell)=
	\langle\phi|P^{(A)}_{I_1}P^{(B)}_{J_1}\cdots P^{(A)}_{I_N}P^{(B)}_{J_N}\psi\rangle=:
	\nu^{(N)}_{\phi,\psi}(\times_{\ell=1}^NI_\ell\times J_\ell)\,.
\end{align*}
Unfortunately, as shown in Example \ref{Ex: counterexample to well-definiteness of complex measure}, $\nu^{(N)}_{\phi,\psi}$ does 
not define a measure on $\mathbb{R}^{2N}$.
However, whenever $[P^{(A)},P^{(B)}]=0$, $\nu^{(N)}_{\phi,\psi}$ does define a complex measure on $\mathbb{R}^{2N}$ -- \textit{cf.} 
section \ref{Sec: The sequence of measures nu^N_phi,psi  generally does not admit (projective) limit}.
Therefore one may wonder whether the convergence $\nu^{(N,n)}_{\phi,\psi}\to\nu^{(N)}_{\phi,\psi}$ holds true in the space of
 complex measure $\mathscr{M}(\mathbb{R}^{2N})$, namely if
\begin{align*}
	\lim_{n\to +\infty}\|\nu^{(N,n)}_{\phi,\psi}-\nu^{(N)}_{\phi,\psi}\|=
	\lim_{n\to +\infty}|\nu^{(N,n)}_{\phi,\psi}-\nu^{(N)}_{\phi,\psi}|(\mathbb{R}^{2N})=0\,.
\end{align*}
Unfortunately this is not the case as shown in particular by the following example.
\begin{example}
	{\em Let us consider $\mathsf{H}=L^2([0,1],dx)$ and let $A=B=X$ be the usual {\em position operator} over $[0,1]$ -- \textit{cf.} example 
\ref{Ex: counterexample to well-definiteness of complex measure}.
	We consider the complete orthonormal basis $\{e_k\}_{k\in\mathbb{Z}}$ of $L^2([0,1],dx)$ made by exponentials
 $e_k(x):=e^{2\pi ikx}$ and we set $Q_n:=\sum_{|k|\leq n} |e_k\rangle\langle e_k|$ for $n\in\mathbb{N}$.
	Moreover we choose $\psi=\phi=1_{[0,1]}\in L^2([0,1],dx)$.
	It follows that, for all $I,J\in\mathscr{B}([0,1])$ (where the Borel algebra is referred to the topology on $[0,1]$ form $\mathbb{R}$),
	\begin{align*}
		\nu^{(1)}_{\phi,\psi}(I\times J)&=
		\langle\phi| P^{(X)}_{I\cap J}\psi\rangle=
		\int_0^1 1_{I\cap J}(x)dx\,,\\
		\nu^{(1,n)}_{\phi,\psi}(I\times J)&=
		\langle\phi|P^{(X)}_IQ_nP^{(X)}_J\psi\rangle=
		\int_0^1\int_0^11_I(x)1_J(y)\rho_n(x-y)dxdy\,,
	\end{align*}
	where $	\rho_n(x-y):=\sum_{|k|\leq n}e_k(x-y)$.
	A direct computation yields $\|\nu^{(1)}_{\phi,\psi}\|=|\nu^{(1)}_{\phi,\psi}|([0,1]^2)=1$ while $\|\nu^{(1,n)}_{\phi,\psi}\|=\|\rho_n\|_{L^1([0,1],dx)}$.
	Since the linear map $L^1([0,1],dx)\ni f\mapsto m_f\in\mathscr{M}([0,1]^2)$ -- defined by $m_f(E):=\int_E f(x-y)dxdy$ --
 is a (non-surjective) isometry and both spaces are complete, it follows that convergence of $\nu^{(1,N)}_{\psi,\psi}$ in $\mathscr{M}([0,1]^2)$
 would imply convergence of $\rho_n$ in $L^1([0,1],dx)$ to some function therein.
	However this is impossible since $\rho_n$ weakly converges to the Dirac delta $\delta$ centred at $0$ which 
cannot be represented by a function of $L^1([0,1],dx)$.}\hfill 
$\blacksquare$ 
\end{example}

\section{Towards a physical interpretation}\label{Sec: physical interpretation}
In the particular   case where the Hilbert space $\sH$ is finite dimensional and no technical issues must be addressed, it is interesting to discuss the physical interpretations of formula \eqref{Eq: f(aA+bB) in the finite dimensional case}. 

To this end we recall  that if, as assumed, the Hilbert space is separable with   $\dim(\sH) \neq 2$, the  {\em Gleason theorem} \cite{Gle57,gleasonbook} establishes that  
 quantum states viewed as $\sigma$-additive measures  $\mu$ on $\mathcal{L}(\mathsf{H})$ are one-to-one with  positive, trace-class, unit-trace operators $\rho_\mu : \sH \to \sH$.
{\em Pure states} are the extremal elements $\nu$ of the convex body of states 
 and they are in a  one-to-one correspondence with unit vectors $\psi_\nu$ up to phases and $\nu(P)=||P\psi_\nu||^2$ is the probability that $P\in \cL(\sH)$ is true if the pure  state is  represented by $\psi_\nu$.

\subsection{Noncommutativity and quantum interference terms}
We start the analysis with the following elementary observation.
Since the identity operator can be decomposed either as $I=\sum _{\lambda \in \sigma (A)}P^{(A)}_\lambda$ or
 as $I=\sum _{\mu \in \sigma (B)}P^{(B)}_\mu$,  for any $N\in \bN$ and $\psi \in \sH$ we have 
\begin{equation} \psi =\sum_{\substack{\lambda_1, \ldots ,\lambda _N \in \sigma(A)\\\mu_1,\ldots , \mu_N\in \sigma(B)}} P^{(A)}_{\lambda_1}  P^{(B)}_{\mu_1}\cdots P^{(A)}_{\lambda_N}  P^{(B)}_{\mu_N}\psi. \label{coherent}\end{equation}
Above, the state vector  $\psi$ is decomposed into a {\em coherent superposition} of generally {\em non-mutually orthogonal} vector states.
These vectors would be mutually orthogonal if the PVMs commuted. The presence in (\ref{Eq: f(aA+bB) in the finite dimensional case})
 of these non-orthogonal vectors is the source 
of a {\em quantum interference phenomenon} we are going to illustrate.

The weak version of formula \eqref{Eq: f(aA+bB) in the finite dimensional case} reads as follows
\begin{equation}\label{prototype1}
	\langle \phi| f(aA+bB)\psi \rangle
	=\lim_{N\to +\infty}\sum_{\substack{\lambda_1, \ldots ,\lambda _N \in \sigma(A)\\\mu_1,\ldots ,  \mu_N\in \sigma(B)}}
	f\left( \frac{1}{N}\sum_{j=1}^N (a\lambda_j+ b\mu_j)\right) \langle \phi| P^{(A)}_{\lambda_1}P^{(B)}_{\mu_1}\cdots P^{(A)}_{\lambda_N}  
	P^{(B)}_{\mu_N}\psi \rangle\:,
\end{equation}
for any $\phi, \psi \in \sH$. In particular when $\phi=\psi$ we get
\begin{equation}\label{prototype2}\langle \psi|f(aA+bB)\psi \rangle=\lim_{N\to +\infty}
\sum_{\substack{\lambda_1, \ldots, \lambda _N \in \sigma(A)\\\mu_1,\ldots,  \mu_N\in \sigma(B)}} f\left( \frac{1}{N}\sum_{j=1}^N (a\lambda_j+ b\mu_j)\right) \langle \psi| P^{(A)}_{\lambda_1} 
 P^{(B)}_{\mu_1}\cdots P^{(A)}_{\lambda_N}  
P^{(B)}_{\mu_N}\psi \rangle\:,\end{equation}
As usual, the left-hand side coincides to the {\em expectation value} of the observable $f(aA+bB)$ when the state is 
pure and represented by the unit vector $\psi$.
In particular, if $f=1_{[\alpha, \beta)}$ (in this case 
 (\ref{prototype2}) should be replaced by \eqref{PVMAB3}),
the left-hand side is nothing but  the {\em probability of obtaining an outcome in $[\alpha,\beta)$
when measuring $aA+bB$ over the pure state represented by the unit vector $\psi$}.  

Exploiting (\ref{coherent}) in the expansion the left entry of the inner products, formula \eqref{prototype2} can be rephrased to
\begin{multline}\label{prototype3}
\langle \psi|f(aA+bB)\psi \rangle
=\lim_{N\to +\infty}\sum_{\substack{\lambda_1, \ldots, \lambda _N \in \sigma(A)\\\mu_1,\ldots , \mu_N\in \sigma(B)}} 
f\left( \frac{1}{N} \sum_{j=1}^N (a\lambda_j+ b\mu_j)\right)\\
\sum_{\substack{\tilde{\lambda}_1,\ldots,\tilde{\lambda}_N\in\sigma(A)\\\tilde{\mu}_1,\ldots,\tilde{\mu}_N\in\sigma(B)}}
\langle P^{(A)}_{\tilde\lambda_1} 
 P^{(B)}_{\tilde\mu_1}\cdots P^{(A)}_{\tilde\lambda_N}  
P^{(B)}_{\tilde\mu_N}\psi| P^{(A)}_{\lambda_1} 
 P^{(B)}_{\mu_1}\cdots P^{(A)}_{\lambda_N}  
P^{(B)}_{\mu_N}\psi \rangle\:,\\
=\lim_{N\to +\infty}\sum_{\substack{\lambda_1, \ldots ,\lambda _N \in \sigma(A)\\\mu_1,\ldots , \mu_N\in \sigma(B)}}
 f\left( \frac{1}{N} \sum_{j=1}^N (a\lambda_j+ b\mu_j)\right)\|P^{(A)}_{\lambda_1} 
 P^{(B)}_{\mu_1}\cdots P^{(A)}_{\lambda_N}  
P^{(B)}_{\mu_N}\psi \|^2+\\
+\lim_{N\to +\infty}\sum_{\substack{\lambda_1, \ldots ,\lambda _N \in \sigma(A)\\\mu_1,\ldots , \mu_N\in \sigma(B)}}
 f\left( \frac{1}{N} \sum_{j=1}^N (a\lambda_j+ b\mu_j)\right)
\\ \sum_{\substack{\tilde{\lambda}_1,\ldots ,\tilde{\lambda}_N\in \sigma(A)\\
		\tilde{\mu_1},\ldots , \tilde{\mu}_N\in \sigma(B),\\
		(\overline{\tilde \lambda},\overline{\tilde \mu})\neq (\bar \lambda, \bar \mu)}}\langle P^{(A)}_{\tilde\lambda_1} 
 P^{(B)}_{\tilde\mu_1}\cdots P^{(A)}_{\tilde\lambda_N}  
P^{(B)}_{\tilde\mu_N}\psi| P^{(A)}_{\lambda_1} 
 P^{(B)}_{\mu_1}\cdots P^{(A)}_{\lambda_N}  
P^{(B)}_{\mu_N}\psi \rangle\:,
\end{multline}
where we defined  $(\bar \lambda, \bar \mu):=(\lambda _1,..., \lambda _N, \mu _1,..., \mu _N)$ and $ (\overline{\tilde \lambda},\overline{\tilde \mu}):=(\tilde\lambda _1,..., \tilde\lambda _N, \tilde\mu _1,..., \tilde\mu _N)$.\\
The considered expectation value is therefore tantamount  to  the sum  two distinct  kinds of terms. 
\begin{itemize} 
\item[(1)] The first type admits   to some extent (see (1) Remark \ref{noclassical} below)   a ``classical''  probabilistic interpretation relying on a recursive use
of  {\em Born's rule} and  the
 standard 
 {\em L\"uders-von Neumann's  post-measurement state postulate} (see, e.g., \cite{BC, landsman, Moretti2017}). As is well-known, 
the former postulate states that, if the initial pure states is represented by the unit vector $\psi\in \sH$, the probability that 
the outcome of an elementary YES-NO observable $P\in \cL(\sH)$ is $1$ (YES) is $||P\psi||^2$. The latter postulate says that, if  the outcome  resulted to be 
$1$, then the post measurement state is represented by the normalized vector $||P\psi||^{-1}P\psi$.
As a matter of fact,   the positive number $$\|P^{(A)}_{\lambda_1} 
 P^{(B)}_{\mu_1}\cdots P^{(A)}_{\lambda_N}  
P^{(B)}_{\mu_N}\psi \|^2\in [0,1]$$ can be interpreted as a {\em conditional probability}\footnote{The general notion of {\em quantum conditional probability} is still an open and  controversial topic  of the quantum theory \cite{Redei,QP,GB13}. Our discussion  will however involve very  basic  and straightforward arguments.}: the probability that, given the initial pure state of
 the system described by  the normalized vector $\psi \in \sH$, in a sequence of $2N$ measurements of  
observables $O_j$, $J=1,...,2N$, with  $O_j=B$ for $j$ odd and $O_j=A$ for $j$ even, the 
sequence of outcomes is exactly  $(\mu_N, \lambda _N,..., \mu _1, \lambda _1)$. Let us illustrate this point in some detail.
Given a general observable $O$ in our finite dimensional Hilbert space with (necessarily discrete) spectrum $\sigma(O)$, and a
 normalized vector $\phi\in \sH$, 
we shall denote with the symbol $\bP(O=o|\phi) $ the probability that, given that the (pure) state of the system is represented
 by $\phi$,  the outcome of the measurement of $O$ is the real number $o\in \sigma (O)$, namely
$$ \bP(O=o|\phi)=|| P^{(O)}_o\phi||^2,$$
where $P^{(O)}_o$ is the spectral projector associated to the eigenvalue $o\in \sigma(O)$. Let us denote by $\psi _{\lambda _j, \mu_j, \lambda_{j+1},\dots , \lambda _N, \mu _N}$ and $\psi _{ \mu_j, \lambda_{j+1},\mu_{j+1},\dots , \lambda _N, \mu _N}$ the post-measurement normalized vectors
$$  \psi _{\lambda _j, \mu_j, \lambda_{j+1},\dots , \lambda _N, \mu _N}
:=\frac{P^{(A)}_{\lambda_j}P^{(B)}_{ \mu_j}P^{(A)}_{ \lambda_{j+1}}\cdots P^{(A)}_{ \lambda _N}
P^{(B)}_{ \mu _N}\psi}{\|P^{(A)}_{\lambda _j}P^{(B)}_{ \mu_j}P^{(A)}_{ \lambda_{j+1}}\cdots
 P^{(A)}_{ \lambda _N}P^{(B)}_{ \mu _N}\psi\|},$$ $$ \psi _{ \mu_j, \lambda_{j+1},\mu_{j+1},\dots , \lambda _N, \mu _N}
:=\frac{P^{(B)}_{ \mu_j}P^{(A)}_{ \lambda_{j+1}}P^{(A)}_{ \mu_{j+1}}\cdots P^{(A)}_{ \lambda _N}P^{(B)}_{ \mu _N}\psi}{\|P^{(B)}_{ \mu_j}
P^{(A)}_{ \lambda_{j+1}}P^{(A)}_{ \mu_{j+1}}\cdots P^{(A)}_{ \lambda _N}P^{(B)}_{ \mu _N}\psi\|} ,$$
with $j=1,...,N$.
We then have:
\begin{multline*}
	\|P^{(A)}_{\lambda_1}P^{(B)}_{\mu_1}\cdots P^{(A)}_{\lambda_N}  P^{(B)}_{\mu_N}\psi \|^2
	=\bP(A=\lambda _1|\psi_{ \mu_1, \lambda_{2},\dots , \lambda _N, \mu _N})
	\|P^{(B)}_{\mu_1}\cdots P^{(A)}_{\lambda_N}  P^{(B)}_{\mu_N}\psi \|^2\\
	=\bP(A=\lambda _1|\psi_{ \mu_1, \lambda_{2},\dots , \lambda _N, \mu _N})
	\bP(B=\mu _1|\psi_{ \lambda_{2},\mu_2 ,\dots , \lambda _N, \mu _N})
	\|P^{(A)}_{\lambda_2}P^{(B)}_{\mu_2}\cdots P^{(A)}_{\lambda_N}  P^{(B)}_{\mu_N}\psi \|^2\,.
\end{multline*}
By induction over $N$ we get the wanted result:
\begin{multline}
\|P^{(A)}_{\lambda_1} 
 P^{(B)}_{\mu_1}\cdots P^{(A)}_{\lambda_N}  
P^{(B)}_{\mu_N}\psi \|^2\\
=\Pi_{j=1}^{N-1}\bP(A=\lambda _j|\psi_{ \mu_j, \lambda_{j+1},\dots , \lambda _N, \mu _N})\bP(B=\mu _j|\psi_{ \lambda_{j+1},\mu_{j+1} ,\dots , \lambda _N, \mu _N})\\\cdot
\bP(A=\lambda_N|\psi_{\mu_N})\bP(B=\mu_N|\psi)\:.
\end{multline}
It is now evident that the right-hand side  has the natural interpretation as the probability that, given the initial pure state of
 the system described by  the normalized vector $\psi \in \sH$, in a sequence of $2N$ measurements of  
observables $O_j$, $J=1,...,2N$, with  $O_j=B$ for $j$ odd and $O_j=A$ for $j$ even, the 
sequence of outcomes is exactly  $(\mu_N, \lambda _N,..., \mu _1, \lambda _1)$. \\
\item[(2)]
The remaining  terms appearing in the last line of equation  \eqref{prototype3} are of the form \begin{equation}\label{interferenceterms}
\langle P^{(A)}_{\tilde\lambda_1} 
 P^{(B)}_{\tilde\mu_1}\cdots P^{(A)}_{\tilde\lambda_N}  
P^{(B)}_{\tilde\mu_N}\psi | P^{(A)}_{\lambda_1} 
 P^{(B)}_{\mu_1}\cdots P^{(A)}_{\lambda_N}  
P^{(B)}_{\mu_N}\psi \rangle,\end{equation}with $ (\tilde \lambda_1,\ldots, \tilde \mu_{N})\neq (\lambda_1,\ldots, \mu_N)$,
  are complex numbers which can be interpreted as {\em quantum amplitudes}  or {\em interference} terms arising from the coherent
 decomposition (\ref{coherent}) into non-orthogonal addends.
   In the trivial case where the spectral measures $\{P^{(A)}_\lambda\}_{\lambda \in \sigma(A)}$ and 
$\{P^{(B)}_\mu\}_{\mu \in \sigma(B)}$ commute, it is easy to check that all the interference terms \eqref{interferenceterms} 
vanish as already remarked. In this sense, their presence has an evident  quantum nature.
\end{itemize}
\begin{remark} \label{noclassical} {\em $\null$\\
{\bf (1)} It is worth stressing that also the {\em quantum conditional probability} discussed in (1) 
enjoys different properties than those of the {\em classical conditional probabilty} when the PVMs do not commute. 
 This is because we are here dealing with the quantum notion of probability characterized by {\em Gleason's theorem},
 defined on a orthomodular lattice rather than a $\sigma$-algebra. In particular the classical identity $\bP(U|V)\bP(V) = \bP(V|U)\bP(U)$ has 
not a corresponding  version at quantum level when $U$ and $V$ are replaced for quantum incompatible events. This is reflected from the fact 
that the order in the sequence of measurements $P^{(B)}_{\mu_1}\cdots P^{(A)}_{\lambda_N}  
P^{(B)}_{\mu_N}$ cannot be ignored.\\
{\bf (2)} 
If $\sH$ is finite dimensional, we have the following limit in the operator norm
\begin{align*}
	f(A+B)=
	\lim_{N\to +\infty }
	f_N(A,B)=
	\lim_{N\to +\infty }
	\sum_{\substack{
		\lambda_1,\ldots,\lambda_N\in\sigma(A)\\
		\mu_1,\ldots,\mu_N\in\sigma(B)}}
	f\bigg(\frac{1}{N}\sum_{j=1}^N(\lambda_j+\mu_j)\bigg)
	P^{(A)}_{\lambda_1}P^{(B)}_{\mu_1}\cdots P^{(A)}_{\lambda_N}P^{(B)}_{\mu_N}\,.
\end{align*}
We now provide an expansion of $f_N(A,B)$ in terms of spectral commutators of $A$ and $B$ which highlights the link between 
non-commutativity of spectral measures and appearance of terms of second kind in the right-hand side of (\ref{prototype3}).
For that let $C_{\lambda,\mu}:=[P^{(A)}_\lambda,P^{(B)}_\mu]$ and set $\wp_N(\lambda_1,\ldots,\lambda_N;\mu_1,\ldots,\mu_N):=P^{(A)}_{\lambda_1}P^{(B)}_{\mu_1}\cdots P^{(A)}_{\lambda_N}P^{(B)}_{\mu_N}$.
A direct computation leads to
\begin{align*}
	\wp_N(\lambda_1,\ldots,\lambda_N;\mu_1,\ldots,\mu_N)&:=
	P^{(A)}_{\lambda_1}P^{(B)}_{\mu_1}\cdots P^{(A)}_{\lambda_N}P^{(B)}_{\mu_N}\\&=
	(P^{(B)}_{\mu_1}\delta_{\lambda_1,\lambda_2}+C_{\lambda_1,\mu_1})
	P^{(A)}_{\lambda_2}P^{(B)}_{\mu_2}\cdots P^{(A)}_{\lambda_N}P^{(B)}_{\mu_N}\\&=
	(P^{(B)}_{\mu_1}\delta_{\lambda_1,\lambda_2}+C_{\lambda_1,\mu_1})
	\wp_{N-1}(\lambda_2,\ldots,\lambda_N;\mu_2,\ldots,\mu_N)\\&=
	\prod_{\ell=1}^N(P^{(B)}_{\mu_\ell}\Delta_{\lambda_\ell,\lambda_{\ell+1}}+C_{\lambda_\ell,\mu_\ell})\,,
\end{align*}
where we set $\Delta_{\lambda_\ell,\lambda_{\ell+1}}=\delta_{\lambda_\ell,\lambda_{\ell+1}}$ for $1\leq \ell<N$, $\delta _{i,j}$ denoting the Kronecker symbol, while $\Delta_{\lambda_N,\lambda_{N+1}}=P^{(A)}_{\lambda_N}$.
Notice that the product of the terms $P^{(B)}_{\mu_\ell}\Delta_{\lambda_\ell,\lambda_{\ell+1}}$ provides the contribution arising 
from the classical convolution formula, namely $\prod_{\ell=1}^NP^{(B)}_{\mu_\ell}\Delta_{\lambda_\ell,\lambda_{\ell+1}}=\prod_{\ell=1}^{N-1}\delta_{\lambda_\ell,\lambda_{\ell+1}}\delta_{\mu_\ell,\mu_{\ell+1}}P^{(A)}_{\lambda_1}P^{(B)}_{\mu_1}$.
Products containing a commutator factor $C_{\lambda,\mu}$ represent a ``perturbative" quantum correction to the commutative formula.
Writing the product of spectral projectors as a sum over insertion of commutators we obtain a ``perturbative"
 expansion of $f(A+B)$ whose first terms are given by
\begin{align*}
	&f(A+B)=
	\sum_{\substack{\lambda\in\sigma(A)\\\mu\in\sigma(B)}}
	f(\lambda+\mu)P^{(B)}_\lambda P^{(A)}_\mu
	\\&+
	\lim_{N\to +\infty }
	\sum_{\substack{\lambda_1,\lambda_2\in\sigma(A)\\\mu_1,\mu_2\in\sigma(B)}}
	f\bigg[\frac{1}{N}\bigg(\lambda _1+(N-1)\lambda_2+\mu_1+(N-1)\mu_2\bigg)\bigg]
	C_{\lambda_1,\mu_1}P^{(B)}_{\mu_2}P^{(A)}_{\lambda_2}+
	\\&+
	\lim_{N\to +\infty }
	\sum_{\substack{\lambda_1,\lambda_2\in\sigma(A)\\\mu_1,\mu_2\in\sigma(B)}}
	f\bigg[\frac{1}{N}\bigg((N-1)\lambda_1+\lambda_2+(N-1)\mu_1+\mu_2\bigg)\bigg]
	P^{(B)}_{\mu_1}C_{\lambda_2,\mu_2}+
	\\&+
	\lim_{N\to +\infty }
	\sum_{\substack{\lambda_1,\lambda_2\in\sigma(A)\\\mu_1,\mu_2,\mu_3\in\sigma(B)}}
	\sum_{k=2}^{N-1}f\bigg[\frac{1}{N}\bigg((k-1)\mu_1+\mu_2+(N-k)\mu_3+k\lambda_1+(N-k)\lambda_2\bigg)\bigg]\\&\cdot
	P^{(B)}_{\mu_1}C_{\lambda_1,\mu_2}P^{(B)}_{\mu_3}P^{(A)}_{\lambda_2}+
	\cdots
\end{align*}
}\hfill $\blacksquare$
\end{remark}

 \subsection{Sum over all possible histories}
From now on, we assume $a=b=1$ for the sake of simplicity,  still sticking to the case of a finite dimensional Hilbert space $\sH$.
 Formula (\ref{Eq: f(aA+bB) in the finite dimensional case})
allows us to construct (weakly) the spectral measure of $A+B$ in terms of a particular limiting procedure. 
Let us consider, for any $N\in \bN$ the finite sequence of $2N$ observables $O_1,...,O_{2N}$, where for $j=1,...,N$ 
$O_{2j}=A$ and $O_{2j-1}=B$. It corresponds to a sequence of  alternating measurements of the couple of observables $A$ and $B$.
As discussed above the statistics of these sequences of measurements is related to the complex measures $\nu_{\phi,\psi}^{(A,B,\dots,A,B)}$ on 
$\sigma (A)\times\dots \times\sigma (B)$  defined as 
$$\nu_{\phi,\psi}^{(A,B,\dots,A,B)}(I_1\times J_1\times \cdots I_N\times J_N)=\langle\phi |P^{(A)}_{I_1}P^{(B)}_{J_1} \cdots P^{(A)}_{I_N}P^{(B)}_{J_N} \psi \rangle.$$
We can in this way  associate to any observable $O_n$ a map  $T(O_n)$ on $\Omega :=(\sigma (A)\times\sigma (B))^\bN$ as follows,
$$T(O_{2j})(\lambda _1, \mu_1,\dots, \lambda _n, \mu_n,\dots)=\lambda _j, \qquad  T(O_{2j-1})(\lambda _1, \mu_1,\dots, \lambda _n, \mu_n,\dots)=\mu _j,\qquad j\geq 1.$$
We shall henceforth set $T(O_{2j})\equiv \xi_j$ and $T(O_{2j-1})\equiv \eta_j$, $j\geq 1$  to shorten the notation. 
Formally we have for any $N\geq 1$ and $f_1,..., f_N, g_1, ..., g_N\in C_b(\bR)$
\begin{equation}\label{QMe}\langle\phi | \Pi_{j=1}^Nf_j(A)g_j(B)\psi \rangle=\bE[\Pi_{j=1}^Nf_j(\xi_j)g_j(\eta_j)]\end{equation}
where the {\em expectation value} on the right-hand side  of \eqref{QMe} stands for the integral
\begin{equation}\label{INT-N}\int_{\sigma (A)\times\dots \times\sigma (B)}\Pi_{j=1}^Nf_j(\lambda_j)g_j(\mu_j)d\nu^{(A,B,\dots,A,B)}_{\phi,\psi}(\lambda_1,...,\mu_N).\end{equation}
Referring to a  measure space  $(\Omega, \cF, \nu) $, with $\Omega = (\sigma (A)\times\sigma (B))^\bN$, $\cF $ the $\sigma $ algebra 
generated by the cylinder sets (see appendix \ref{App: projective limit of complex measure}) and $\nu$
 a sigma additive measure on $\cF$ with finite total variation (see  section 3.2),  the sets $\{\xi_n\}_{n\in \mathbb{N}}$ and
  $\{\eta_n\}_{n \in \mathbb{N}}$ cannot in
 general be regarded as sequences of {\em random variables} (i.e. measurable functions) thereon, where $\nu$ extends $\nu^{(A,B,\dots,A,B)}_{\phi,\psi}$.
However, for any finite $N$  the integral \eqref{INT-N} is well defined and  $\{\xi_n\}_{n\in \mathbb{N}}$ and  $\{\eta_n\}_{n \in \mathbb{N}}$
 are called {\em pseudo stochastic process}. They find applications, e.g., in the
 mathematical definition of {\em Feynman path integrals} as well as in the construction of {\em Feynman-Kac type} formulae for
 PDEs that do not satisfy maximum principle \cite{AlMa16}.

In this context, the construction so far defined  describes the  spectral measure of $A+B$ in terms of a particular limiting procedure.
Referring to the sequence of pseudo-random variables $\{\zeta_n\}_{n \in \mathbb{N}}$ defined as $\zeta_n:=\xi_n+\eta_n$,  
for $N\geq 1$ let $S_N$ denote the map 
$S_N:=\sum_{n=1}^N\zeta_n$. Then, according to formula \eqref{Eq: f(aA+bB) in the finite dimensional case}, for any
 function $f\in \cF(\bR)$  (or also $f \in C^\infty(\bR)$  since here
 $A,B \in \gB(\sH)$ automatically because $\dim \sH < +\infty$), it holds
$$\langle \phi | f(A+B)\psi\rangle =\lim_{N\to +\infty}\bE\left[f\left(\frac{S_N}{N}\right)\right]\:,$$
which can be interpreted as a weak convergence of the (complex) distribution of the arithmetic mean of the 
first $N$ pseudorandom variables $\zeta_n$ 
to the measure $\cB(\bR) \ni E \mapsto \langle \phi | P^{(A+B)}_{E}\psi\rangle$. Since we are dealing with  discrete spectra, the 
formula above can be rephrased to
$$\langle \phi | f(A+B)\psi\rangle =\lim_{N\to +\infty}\sum_{\substack{\lambda_1,\ldots,\lambda_N\in\sigma(A)\\\mu_1,\ldots,\mu_N\in\sigma(B)}}
f\left(\frac{1}{N}\sum_{n=1}^N(\lambda _n+\mu_n)\right)\nu_{\phi,\psi}^{(A,B,...,A,B)}(\{\lambda _1, \mu_1,\dots, \lambda_N, \mu_N\})$$
and the right-hand side can be interpreted as a ``sum over all possible histories''. In other words that is an integral over 
the set of all possible outcomes of a sequence of consecutive measurements of the observables $A$ and $B$.
It is worth stressing that the same result can be obtained by means of a different 
construction, by considering the family of complex measures $\nu_{\phi,\psi}^{(B,A,...,B,A)}$
 instead of $\nu_{\phi,\psi}^{(A,B,...,A,B)}$, describing a sequence of consecutive
 measurement of the observables $A$ and $B$ which starts with the measurement of $A$.

\section{Summary} 
This work tackles the issue of connecting the spectral measure $P^{(aA+bB)}$ of a linear combination of two  observables $A,B$
with the spectral measures $P^{(A)},P^{(B)}$ of $A$ and $B$,  in the general  case of a pair of non-compatible observables. This problem is  relevant from the physical  point of view  as it helps in justifying  the the  linear structure of the space of observables,  
which is assumed in the standard formulation of Quantum Theory.   As a matter of fact, in Section \ref{Sec: extension to infinite dimensional Hilbert spaces} we established  a new  formula (\ref{mainid2}), relating suitable  functions $f(aA+bB)$ of two generally non-commuting and unbounded  observables $A$ and $B$ with their PVMs  $P^{(A)}$ and $P^{(B)}$. That formula  is the central achievement of this paper 
expressed by Theorem \ref{teorapgen}, whose proof relies  essentially  on Trotter formula.
This constrains the class of functions $f$ admitted by (\ref{mainid2}) -- in Section \ref{Sec: counterexamples} we show that the hypothesis necessary to establish our main result  cannot be relaxed -- however, this class of functions is nevertheless sufficiently large to allow the  application of  (\ref{mainid2}) to various cases of physical relevance discussed in Section \ref{Sec: some applications}.
In fact, our formula  has a nice interplay with the structure of Jordan algebra of observables and also with  the notion of Feynman integral and the Feynman-Kac formula since, essentially,  it specializes to them when $A$ and $B$ are suitably chosen.
In addition, identity  (\ref{mainid2}) has an appealing aspect from the physical viewpoint  as it contains some terms which are suitable for some operational physical intepretation, as discussed in 
Section 
\ref{Sec: physical interpretation}.
   \\

\noindent {\bf Acknowledgments}.  The authors are grateful to C. Fewster, R. Ghiloni, and I. Khavkine for useful discussions.
We also thank an anonymous reviewer for his/her detailed and careful  suggestions to improve the text and for having pointed out  further relevant bibliography.

\appendix

\section{Proof of some propositions}\label{App: proof of some propositions}

{\bf Proof of Lemma \ref{proptec}}.  
(a) First of all, let us prove that  there is a complex Borel measure on $\bR^{2N}$ (with finite total variation according to  Section \ref{Sec: introduction}) 
satisfying  \eqref{reg-measures}. Notice that, if such a measure exists, it must be unique  in view of  \cite[Theorem 3.3]{Bill}.
Let $\{h_k\}_{k=1,..,d_n}$ be an orthonormal basis of the finite-dimensional subspace $Q_n(\sH)$, where $d_n=\mbox{dim } (Q_n(\sH))$.  Then 
\begin{align*}
	\nu_{\phi,\psi,Q}^{(N,n)}(\times_{\ell=1}^nI_\ell\times J_\ell)&:=
	\langle\phi |P^{(A)}_{I_1}Q_{n}P^{(B)}_{J_1}\cdots P^{(A)}_{I_N}Q_{n}P^{(B)}_{J_N}\psi\rangle
	\\&=
	\sum_{k_1,\ldots,k_{2N-1}\leq d_n}\langle\phi | P^{(A)}_{I_1}h_{k_1}\rangle \langle h_{k_1}|P^{(B)}_{J_1}h_{k_2}\rangle \cdots  \langle h_{k_{2N-1}}|P^{(B)}_{J_N}\psi\rangle 
\end{align*}
 Let us consider the unitary maps $U^{(A)}\colon\sH\to L^2(\bN \times \bR,\nu^{(A)})$, $U^{(B)}\colon\sH\to L^2(\bN \times \bR,\nu^{(B)})$ defined by the spectral representation of the self-adjoint operators $A,B$ as in Theorem \ref{teorap}.
Setting $h^{(A)}_j(r):=(U^{(A)}h_j)(r)$ and $h^{(B)}_j(s):=(U^{(B)}h_j)(s)$, where $r,s\in \Gamma := \bN \times \bR$ the latter equipped with the 
product measure when the topology on $\bN$ is the discrete one,  we have, 
\begin{align*}
	&\nu^{(N,n)}_{\phi,\psi,Q}(\times_{\ell=1}^nI_\ell\times J_\ell)=
	\sum_{k_1,\ldots,k_{2N-1}\leq d_n}\langle \phi^{(A)}|1_{I_1}h^{(A)}_{k_1}\rangle \langle h^{(B)}_{k_1}|1_{J_1}h^{(B)}_{k_2}\rangle \cdots\langle h^{(B)}_{k_{2N-1}}|1_{J_N}\psi^{(B)}\rangle\\&=
	\int _{(\bN \times \bR)^{2N}} 1_{\bN \times I_1 \times \cdots \times \bN \times J_N}(r_1, \ldots, s_N)
\rho^{(n)}_{\phi,\psi,Q}(r_1, \ldots ,s_N)\mathrm{d}\nu^{(A)}(r_1)\cdots\mathrm{d}\nu^{(B)}(s_N)\,,
\end{align*}
where we defined $\phi^{(A)}:=U^{(A)}\phi$, $\psi^{(B)}:=U^{(B)}\psi$ and
\begin{align*}
	&\rho^{(n)}_{\phi,\psi,Q}(r_1, \ldots , s_N):=\sum_{k_1,\ldots,k_{2N-1}\leq d_n}
	\overline{\phi^{(A)}(r_1)}h^{(A)}_{k_1}(r_1)\overline{h^{(B)}_{k_1}(s_1)}\cdots\overline{h^{(B)}_{k_{2N-1}}(s_N)}\psi^{(B)}(s_N)\,.
\end{align*}
Each factor on the right-hand side (like 
$\overline{\phi^{(A)}(r_1)}h^{(A)}_{k_1}(r_1)$) is  made of a product of a pair of $L^2$ functions with the same argument ($r_1$), therefore it is an element of $L^1$ with respect to the corresponding measure $\nu^{(A)}$.
Therefore  $\rho^{(n)}_{\phi,\psi,Q}\in L^1((\bN \times \bR)^{2N}, (\nu^{(A)}\otimes\nu^{(B)})^{\otimes N})$ because it is a finite
sum of $L^1$ functions in the product measure. Hence, the set function $\mu^{(N,n)}_{\phi,\psi,Q}:\cB (\Gamma^{2N})\to \bC$ defined as
\begin{equation}\mu^{(N,n)}_{\phi,\psi,Q}(F):=\int _F \rho^{(n)}_{\phi,\psi,Q}(r_1, \ldots , s_N)\mathrm{d}\nu^{(A)}(r_1)\cdots\mathrm{d}\nu^{(B)}(s_N)\:, \quad F\in \cB(\Gamma^{2N})\:,\label{mu}\end{equation}
is a $\sigma$-additive complex measure with total variation
$|\mu^{(N,n)}_{\phi,\psi,Q}|(\Gamma^{2N}) \leq \|\rho^{(n)}_{\phi,\psi,Q}\|_{L^1}$. Let  $\Pi:\Gamma ^{2N}\to \bR^{2N}$ be the surjective  map defined by
$$\Pi(m_1, \lambda_1, \ldots, m_N ,\lambda_N, l_1, \mu_1, \ldots,  l_N, \mu_N):=( \lambda_1, \ldots,  \lambda_N,  \mu_1, \ldots,  \mu_N),$$
where $m_1,\ldots,m_N,l_1,\ldots,l_N\in\mathbb{N}$ while $\lambda_1,\ldots,\lambda_N,\mu_1,\ldots,\mu_N\in\mathbb{R}$.
Clearly $\Pi$ is measurable because it is continuous with respect to the natural product topologies.
We then define the complex measure $\nu^{(N,n)}_{\phi,\psi,Q} $ 
on $\cB(\bR^{2N})$ extending the values $\nu_{\phi,\psi,Q}^{(N,n)}(\times_{\ell=1}^nI_\ell\times J_\ell)$
as the {\em image measure} (also called {\em pushforward measure}) of $\mu^{(N,n)}_{\phi,\psi,Q}$, namely $\nu^{(N,n)}_{\phi,\psi,Q}=(\Pi)_*(\mu^{(N,n)}_{\phi,\psi,Q}) := 
\mu^{(N,n)}_{\phi,\psi,Q} \circ \Pi^{-1}$. With the said definition, if $O \subset \bR^{2N} \setminus \sigma(A) \times \cdots \times \sigma(B)$
 is open,  then $\Pi^{-1}(O)$ is open and stays in the complement to  the set $\bN\times \sigma(A) \times \cdots \times \bN\times \sigma(B)$. 
This set includes $\mbox{supp}((\nu^{(A)}\otimes\nu^{(B)})^{\otimes N})$ according to (c) Theorem \ref{teorap}. 
Hence 
$(\nu^{(A)}\otimes\nu^{(B)})^{\otimes N}(\Pi^{-1}(O)) =0$. Using (\ref{mu}) observing that $(\nu^{(A)}\otimes\nu^{(B)})^{\otimes N}$ is a 
positive measure and $\rho^{(n)}_{\phi,\psi,Q}$ is $L^1$ with respect to that measure, we conclude that $|\mu^{(N,n)}_{\phi,\psi,Q}|(\Pi^{-1}(O)) =0$ 
which implies  $|\nu^{(N,n)}_{\phi,\psi,Q}|(O)=0$. In other words, (\ref{supp-reg-measures}) is true.

(b) Let us define
\begin{equation} I^{(N)}_{\phi,\psi }(f):=\int_\bR\langle \phi |  (e^{i \frac{t}{N}aA}e^{i \frac{t}{N}bB})^N\psi\rangle d\nu _f(t)\:,\label{IN}\end{equation}
where $\nu_f = F^{-1}(f)$ referring to (\ref{Eqn: Fourier transform for measure}).
An elementary inductive argument  establishes that $$S_1Q_nS_2Q_n \cdots Q_n S_N \to S_1S_2\cdots S_N\quad \mbox{strongly as $n\to +\infty$,}$$  
 if (1) the orthogonal projectors $Q_n$ satisfy $Q_n \to I$ in the strong sense and (2) $S_1, \ldots, S_N \in \gB(\sH)$ satisfy $\|S_j\| \leq c <+\infty$ for $j=1,\ldots,N$.
Taking advantage of this fact (noticing that $e^{i s C}$ is unitary if $s\in \bR$ and $C$ is selfadjoint, so that $\|e^{i s C}\|=1$),
Lebesgue's dominated convergence theorem and  the Fubini theorem allow us to write
 \begin{align}
I^{(N)}_{\phi,\psi }(f)&=\int_\bR \langle \phi|  e^{i \frac{t}{N}aA}e^{i \frac{t}{N}bB}\cdots e^{i \frac{t}{N}aA}e^{i \frac{t}{N}bB} \psi\rangle d\nu _f(t)\nonumber \\
&=\int_\bR \langle \phi|  \lim_{n\to \infty} e^{i \frac{t}{N}aA}Q_ne^{i \frac{t}{N}bB}\cdots e^{i \frac{t}{N}aA}Q_ne^{i \frac{t}{N}bB} \psi\rangle d\nu _f(t)\nonumber \\
&=\lim_{n\to \infty}  \int_\bR\langle \phi| e^{i \frac{t}{N}aA}Q_ne^{i \frac{t}{N}bB}\cdots e^{i \frac{t}{N}aA}Q_ne^{i \frac{t}{N}bB} \psi\rangle d\nu _f(t)\nonumber \\
&=\lim_{n\to \infty}  \int_\bR\int_{\bR^{2N}}e^{i\frac{t}{N}\sum_{l=1}^N(a\lambda _l+b\mu_l)}d\nu^{N,n}_{\phi,\psi,Q}(\lambda _1, \mu_1,\dots , \lambda _N, \mu_N) d\nu _f(t)\nonumber \\
    &=\lim_{n\to \infty} \int_{\bR^{2N}}f\left(\frac{1}{N}\sum_{l=1}^N(a\lambda _l+b\mu_l)\right)d\nu^{(N,n)}_{\phi,\psi,Q}(\lambda _1, \mu_1,\dots , \lambda _N, \mu_N) \label{limIN}\:.
\end{align}
Due to (\ref{IN}), the limit does not depend on the chosen sequence $Q$ as declared in the hypothesis. 
\hfill $\Box$

\section{Spectral representation theorem}\label{App: spectral representation theorem}
We state and prove here a version of  the spectral representation theorem for selfadjoint operators we exploited  in the proof of Lemma \ref{proptec}. This is nothing but a refinement of part of some results
which can be found in \cite{RS1,Moretti2017}.

\begin{theorem}\label{teorap}
	Let $T: D(T) \to \sH$ a selfadjoint operator in the  separable   Hilbert space $\sH$.
	The following facts are true.
	\begin{itemize}
		\item[(a)]
		There is a  measure space $(\Gamma, \Sigma, \nu)$, 
		where 
		\begin{itemize} \item[(i)] $\Gamma := \bN \times \bR$,
			\item[(ii)] $\Sigma:= \cB(\Gamma)$ referred to the product topology of $\bN \times \bR$, equipping $\bN$ with the discrete topology,
			\item[(iii)] $\nu$ is a positive $\sigma$-additive Borel measure over $\Gamma$ which  can be chosen to be  finite.
		\end{itemize}
		\item[(b)] 
		There is a unitary map $U: \sH \to L^2(\Gamma, \Sigma, \nu)$,
		such that the PVM $P^{(T)}$ of $T$ satisfies
		\begin{equation} \left(U P^{(T)}_E \psi\right)(n,x) =1_E(x)  (U\psi)(n,x)\quad \mbox{for every $E\in \cB(\bR)$, $(n,x)\in \Gamma$.}\label{fine}\end{equation}
		\item[(c)]
		It holds $\mbox{supp}(\nu) \subset \bN \times \sigma(T)$.
	\end{itemize}
\end{theorem}

\begin{proof} (a) and (b). If $\sH=\{0\}$ everything is trivial, so that we assume this is not the case. Take $\psi_0 \in \sH\setminus \{0\}$ and
 define the  subspace $\sH_0 := \{f(T)\psi_0 \:|\: f \in L^2(\bR, \nu_{\psi_0})\}$ where henceforth $\nu_\psi(E):= \langle \psi|P^{(T)}_E\psi\rangle = \|P^{(T)}_E\psi\|^2$.  The linear map 
$$V_0 :  L^2(\bR, \nu_{\psi_0}) \ni f \mapsto f(T)\psi_0\in  \sH_0$$ is evidently isometric because $\|f(T) \psi_0\|^2 = \|f\|^2_{L^2(\bR, \nu_{\psi_0})}$
 from the elementary spectral theory \cite{Moretti2017}, and surjective by construction. Furthermore $\sH_0$ is closed:
  if $\sH \ni \psi = \lim_{n\to +\infty}f_n(T)\psi_0$, then $\{f_n\}_{n\in \bN}$ is Cauchy in $L^2(\bR, \nu_{\psi_0})$ and 
its limit $f\in L^2(\bR, \nu_{\psi_0})$ satisfies $f(T)\psi_0 = \psi$.
Now take $\psi_1 \in \sH_0^\perp$ with $\psi_1 \neq 0$, 
if there is such vector, and construct an analogous isometric map $$V_1 :  L^2(\bR, \nu_{\psi_1}) \ni f \mapsto f(T)\psi_1\in  \sH_1$$
observing that $\sH_1 \perp \sH_0$, because $\langle f(T)\psi_1|g(T)\psi_0 \rangle = \int_{\bR} g(\lambda) d\nu_{f(T)\psi_1,\psi_0}(\lambda)$, 
but again  from the elementary spectral theory,
 $$\nu_{f(T)\psi_1,\psi_0}(E) := \langle f(T) \psi_1|P_E \psi_0\rangle =  \langle f(T) \psi_1|1_E(T) \psi_0\rangle =\langle \psi_1|(\overline{f}\cdot1_E)(T) \psi_0\rangle = 0\:.$$
We can iterate the procedure finding a class of mutually orthogonal closed subspaces $\{\sH_j\}_{j\in \bN}$ and a corresponding
 sequence of isometric surjective maps $V_j :  L^2(\bR, \nu_{\psi_j}) \to \sH_j$, where possibly $\sH_{j} = \{0\}$ and $\nu_{\psi_{j}}=0$  from some $j_0$ on.
We assume that $\psi_j \neq 0$ for all $j\in \bN$ since the other case is trivial.
Since $\sH$ is separable, with a  straightforward application of Zorn's lemma,  we have that the countable  Hilbert orthogonal decomposition holds $\sH = \oplus_{j=0}^{+\infty} \sH_j$.  As a consequence, $\sH$ turns out to be isomorphic to 
$\oplus_{j=0}^{+\infty} L^2(\bR, d\nu_{\psi_j})$ and the isomorphism is nothing but $U' := \oplus_{j=0}^{+\infty} V_j^{-1}:
\oplus_{j=0}^{+\infty} \sH_j \to \oplus_{j=0}^{+\infty}L^2(\bR, d\nu_{\psi_j})\:.$
Notice that,  taking advantage of the spectral identity $P_E^{(T)} = 1_E(T)$,
we have
\begin{equation}  U'P_E^{(T)}\psi = \oplus_{j=0}^{+\infty} 1_E\cdot  (V_j^{-1} Q_j\psi) \label{fine2}\:.
\end{equation}
where $Q_j : \sH \to \sH$ is the orthogonal projector onto $\sH_j$.
At this juncture, we are free to change the normalization of the vectors $\psi_j$ without affecting the found subspaces $\sH_j$, 
requiring that  $\|\psi_j\|^2 = 2^{-j}$. With the given definitions, if $\Gamma := \bN \times \bR$,  $\Sigma: = \cB(\bN \times \bR)$,
 and  $\bR_n := \{n\}\times \bR$, we construct  the Borel measure over $\Gamma$
$$\nu(\cup_n E_n):= \sum_{n=0}^{+\infty}\nu_{\psi_n}(E_n)\:, \quad \mbox{for $E_n \in \cB(\bR_n)\equiv \cB(\bR)$}$$
Observe thar every $E \in \cB(\bN \times \bR)$ uniquely decomposes $E=\cup_n E_n$ with $E_n \in \cB(\bR_n)$, since $\bR_n \cap \bR_m =\emptyset$ for $n \neq m$. Noticing that all measures $\nu_{\psi_n}$ are positive and $\sigma$-additive, it is not difficult to establish that $\nu$ is such. In particular $\nu$ is finite if we choose $\|\psi_j\|^2 = 2^{-j}$,  because
$\nu(\Gamma) = \sum_{n=0}^{+\infty} 2^{-n} = 2< +\infty$. Finally, it is evident that 
$f \in L^2(\Gamma, \Sigma, \nu)$ if and only if $f(n, \cdot ) \in L^2(\bR, \nu_{\psi_n})$ for all $n \in \bN$ and
$\sum_{n=0}^{+\infty}\|f(n, \cdot)\|^2_{L^2(\bR, \nu_{\psi_n})} <+\infty$.
In particular 
$$\|f\|^2_{L^2(\Gamma, \Sigma, \nu)} = \sum_{n=0}^{+\infty}\|f(n, \cdot)\|^2_{L^2(\bR, \nu_{\psi_n})} \:,$$
so that there is a unitary map 
$U'' : \oplus_{j=0}^{+\infty} L^2(\bR, d\nu_{\psi_j}) \to L^2(\Gamma, \Sigma, \nu)\:,$
obviously constructed. The wanted Hilbert-space isomorphism $U$ is therefore $$U := U''U' : \sH \to L^2(\Gamma, \Sigma, \nu)\:.$$
With this definition, (\ref{fine}) is a direct rephrasing   of (\ref{fine2}).  \\
(c) It easily arises from the standard result $\mbox{supp}(P^{(T)})= \sigma(T)$, using $\mbox{supp}(\nu_\psi) \subset \mbox{supp}(P^{(T)})$ and the disjoint 
decomposition of $\Gamma$ into open-closed subsets $\bR_n$.
\end{proof}
\section{Projective limit of complex measures}\label{App: projective limit of complex measure}
We state here a particular case of a general result extensively discussed in \cite{Tho,AlMa16}, concerning an extension to complex
 measures of Kolmogorov  theorem for the existence of a so-called {\em consistent family} of probability measures (see, e.g., \cite{Pro,Xia,Yam,Yeh}).
Let $\bN$ be the index set and let us denote with $\cF_\bN$ the collection of subsets $J\subset \bN$ with a finite number of elements 
and with $|J|$ the cardinality of $J\in \cF_\bN$.  \\
$\cF_\bN$ inherits naturally the structure of a directed set, where the partial order relation $\leq$ is defined as $J\leq K$ iff $J\subset K$,  $J,K \in \cF_\bN$. Let us associate  to  any $J\in \cF_\bN$ a measure space $\bR^J$, naturally isomorphic to $\bR^{|J|}$ and let $\Sigma_J$ denote its  Borel $\sigma$-algebra. For any pair $J,K \in \cF_\bN$ with $J\leq K$ let $\pi^K_J:\bR^K\to \bR^J$ be the projection map, which is continuous hence Borel measurable.  \\
A collections $\{\mu_J\}_{J\in  \cF_\bN}$ of complex measures $\mu_J:\Sigma _J\to \bC$ is said to be {\bf consistent } or 
{\bf projective} if the following compatibility condition is fulfilled for any $J,K \in \cF_\bN$, $J\leq K$:
\begin{equation}\label{consistent}
\mu_J=(\pi^K_J)_*\mu_K,
\end{equation}
where $(\pi^K_J)_*\mu_K$ denotes the image (pushforward) measure of $\mu_K$ under the action of the map $\pi^K_J$, namely
$(\pi^K_J)_*\mu_K (E):=\mu_K\left( (\pi^K_J)^{-1}(E)\right)$, for all $E\in \Sigma_J$.\\
Let $\bR^\bN$ be the space of all real-valued sequences and for any $J\in \cF_\bN$ let $\pi_J:\bR^\bN\to \bR^J$  indicate  the projection map. It is 
simple to check that  for any $J,K \in \cF_\bN$, with $J\leq K$ the following composition property holds:
\begin{equation}\label{composition}
\pi_J=\pi^K_J\circ \pi_K.
\end{equation}
We shall consider the $\sigma $-algebra $\Sigma$ on $\bR^\bN$ generated by the {\bf  cylinder sets}, i.e., the subsets of the form $(\pi_J)^{-1}(E)$
 for some $J\in \cF_\bN$ and $E\in \Sigma_J$.  It is simple to see that, for  every  complex measure $\mu:\Sigma\to \bC$,  it is possible to define a collections $\{\mu_J\}_{J\in  \cF_\bN}$ of complex measures $\mu_J:\Sigma _J\to \bC$ as the pushforward of $\mu$ under the action of the projections $\pi_j$, i.e. 
\begin{equation}\label{existencemu} \mu_J:=(\pi_J)_*\mu.\end{equation}
 Further, by the composition property \eqref{composition}, the resulting collection of measures $\{\mu_J\}_{J\in  \cF_\bN}$ is 
consistent, i.e. satisfies condition \eqref{consistent} which is indeed a necessary condition for the existence of a measure $\mu:\Sigma \to \bC$
 generating the family  $\{\mu_J\}_{J\in  \cF_\bN}$ through \eqref{existencemu}. Remarkably, in the case where the measures $\{\mu_J\}_{J\in  \cF_\bN}$ are probability measures, according to Kolmogorov existence theorem 
\cite{Pro,Xia,Yam,Yeh} the consistency condition \eqref{consistent} is also sufficient for the existence of a probability measure $\mu :\Sigma \to \bR$ such that \eqref{existencemu} holds. This result is a cornerstone of probability theory, actually providing the existence of {\em stochastic processes}. 
 In the case where the measures $\{\mu_J\}_{J\in  \cF_\bN}$ belonging to the projective family are complex (or signed) measures,
 the classical Kolmogorov theorem can no longer be directly applied. Indeed, in this case,  additional conditions have to be required and we 
refer to \cite{Tho} for details. In particular, a fundamental condition is the following upper bound on the total variation of the measures 
belonging to the family $\{\mu_J\}_{J\in  \cF_\bN}$:
 \begin{equation}\label{upperbound}
 \sup_J\|\mu_J\|<\infty\,,\qquad\|\mu_J\|:=|\mu_J|(\mathbb{R}^J)\,,
 \end{equation}
 Indeed, the necessity of condition \eqref{upperbound} can be simply proved by observing that whenever $h:X\to Y$
 is a measurable map between two  measurable spaces $(X, \Sigma _X)$ and $(Y, \Sigma _Y)$ and for any complex measure $\mu $ 
on $(X, \Sigma _X)$, the total variation of  the pushforward measure $(h)_*\mu$ on $(Y, \Sigma _Y)$ satisfies the following inequality \footnote{ 
 For any partition $\{E_j\}_j\subset \Sigma _Y$ of  $Y$, we have that the collection of sets  $\{h^{-1}(E_j)\}_j\subset \Sigma _X$ is a partition of $X$. Hence 
 $$\sum_j|(h)_*\mu (E_j)|=\sum_j|\mu (h^{-1}(E_j))|\leq \|\mu\|.$$}
 \begin{equation}\label{inequalitytotalvariation}
 \|(h)_*\mu\|\leq \|\mu\|
 \end{equation}
 Hence, by equation \eqref{existencemu}, for any $J\in \cF_\bN$ we obtain $$\|\mu_J\|=\|(\pi_J)_*\mu\|\leq \|\mu\|\,,$$ which yields \eqref{upperbound}.

\end{document}